\documentclass[10pt,twocolumn,twoside]{IEEEtran}

\ifCLASSINFOpdf
\else
\fi
\usepackage{lipsum}
\usepackage{mathtools}
\usepackage{cuted}
\usepackage{stfloats}
\usepackage{amsmath}
\usepackage{amsfonts}
\usepackage{amssymb}
\usepackage[dvips]{graphicx}
\usepackage{booktabs}
\usepackage{verbatim}

\usepackage{booktabs}
\usepackage{amsmath}
\usepackage[T1]{fontenc}
\usepackage{graphicx}
\usepackage{color}
\usepackage{epstopdf} 
\usepackage{amsthm}
\usepackage{amssymb}
\usepackage{graphicx}
\usepackage{booktabs}
\usepackage{gensymb}
\usepackage{epstopdf}
\usepackage{float}
\usepackage{placeins}
\usepackage{textcomp}
\usepackage{etoolbox}
\usepackage{enumitem}
\usepackage{cite}
\usepackage{url}
\usepackage{multirow}
 \newtheorem{assumption}{Assumption}
\newcommand{\EQ}{\begin{eqnarray}\begin{aligned}}
\newcommand{\EN}{\end{aligned}\end{eqnarray}}
\newcommand{\EQQ}{\begin{eqnarray*}\begin{aligned}}
\newcommand{\ENN}{\end{aligned}\end{eqnarray*}}

\newtheorem{theorem}{\bf Theorem}[section]

\newtheorem{lemma}{\bf Lemma}[section]

\newtheorem{remark}{\bf Remark}[section]

\renewcommand{\t}{^{\mbox{\tiny {\rm T}}}}
\renewcommand{\c}{^{\rm c}}
\renewcommand{\r}{^{\rm r}}

\ifCLASSINFOpdf
\else
\fi

\begin{document}

\title{Event-Triggered Output Synchronization of Heterogeneous Nonlinear  Multi-Agents}

\author{Gulam Dastagir Khan, Zhiyong Chen, and Yamin Yan
\thanks{The authors are with the School of Electrical Engineering and Computing, University of Newcastle, Callaghan, NSW 2308, Australia. {Email-\texttt{ GulamDastagir.Khan@uon.edu.au}, \{\texttt{zhiyong.chen, yamin.yan}\} \texttt{@newcastle.edu.au}. 
Z. Chen is the corresponding author. Tel: +61 2 49216352, Fax: +61 2 49216993. 
}}}

\maketitle
  
\begin{abstract}
This paper addresses the output synchronization problem for heterogeneous nonlinear multi-agent systems with distributed event-based controllers. Employing the two-step synchronization process, we first outline the distributed event-triggered consensus controllers for linear reference models under a directed communication topology. It is further shown that the subsequent triggering instants are based on intermittent communication. Secondly, by using
certain input-to-state stability (ISS) property, we design an event-triggered perturbed output regulation controller for each nonlinear multi-agent. The ISS technique used in this paper is based on the milder condition that each agent has a certain ISS property from input (actuator) disturbance to state rather than
measurement (sensor) disturbance to state. With the two-step design, the objective of output synchronization is successfully achieved with Zeno behavior avoided.

\end{abstract}
\begin{IEEEkeywords} Event-triggered control, multi-agent systems, 
output regulation, nonlinear systems, robust control, Zeno behavior
\end{IEEEkeywords}

\section{Introduction}

Over the past years, researchers in the field of systems and control have made extensive efforts in designing various event-triggering algorithms for networked multi-agent systems (MASs) for the tasks like consensus, synchronization, formation, etc. The amalgamation of event-triggering techniques with MASs increases the functionality of resource-limited microprocessors in a networked control environment.  


In this paper, we study the event-triggered robust output synchronization problem under the two-step framework proposed in \cite{chen2016robust,Zhu2016}. This means that the overall synchronization process is broken down into two steps or processes. The first step is the event-based consensus of linear homogeneous reference models (exosystems) and the second one is the event-based robust perturbed output regulation of a nonlinear system.  Since these two processes occur concurrently, from the regulation perspective, the network influences are regarded as external perturbation to an agent's exosystem, eventually resulting in robust perturbed output regulation problem.   

The aforementioned two steps involve two technical problems that, combined, achieve the final objective of output synchronization
of nonlinear heterogeneous systems in an event-triggered control fashion. Even though the integration of these two problems 
is novel in the event-triggered control scenario, the solutions to these two problems have their own independent contributions
to be elaborated as follows.  

\medskip
\noindent {\it  Problem 1- Event-triggered Reference Model Consensus}

The literature on event-triggered consensus of linear MASs is rich. However, a complete solution is still
inadequate for some fundamental technical bottlenecks. The early work can be found in 
\cite{dimarogonas2012distributed,dimarogonas2009aevent,dimarogonas2009bevent} that requires 
continuous communication between neighboring agents.  
The disadvantage of continuous communication is obvious, which contradicts the philosophy of event-triggered control.
All the later development was mainly on the triggering mechanism not relying on continuous communication. 
A conceivable scenario is that each agent updates its control action and broadcasts its state when 
its own event occurs, which depends on the neighbors' broadcast states. 
{\it A fundamental difficulty has arisen that is the tradeoff between the lower bound of the 
inter-event intervals (reciprocal of event frequency) and steady-state control error.}
A smaller control error may result in smaller inter-event intervals.
For a triggering system without a strictly positive lower bound for
inter-event intervals,  an infinite number of events may occur in a finite time period, called 
Zeno behavior.  Roughly speaking, to exclude the risk of Zeno behavior, 
the design cost is the loss of asymptotic control accuracy. 
 
 In the early work \cite{aastrom1999comparison}, the idea is to trigger an event 
whenever the state deviates from the equilibrium by a specified constant threshold. 
It excludes Zeno behavior, but makes the system state converge to a ball centered at the origin, 
not asymptotically to the origin, due to the aforementioned tradeoff. 
It is worth mentioning that the event function, a function that characterizes the conditions to trigger an event, 
used in \cite{aastrom1999comparison} is state independent. 
The similar idea of inclusion of a constant threshold in a state dependent event function was discussed 
in \cite{borgers2014event}. 
Different strategies of adding a threshold in an event function for excluding 
Zeno behavior can be found in other works on event-triggered control of MASs in 
\cite{garcia2013decentralised,zhu2014event,xing2017event}. 
 
In recent five years, researchers have put great efforts in solving the tradeoff, 
targeting both exclusion of Zeno behavior and asymptotic control performance without a steady-state error. 
To the best of our knowledge, there has not been a success in achieving both targets,  unless some additional cost is applied. 
Before we discuss the solutions of achieving both targets with some additional costs, we would like to mention 
some results that claim both without those costs. But unfortunately, we believe there exist some flaws in the technical development. 
The ``no free lunch'' theorem still applies for the research on this topic.   
The technical flaw in proving exclusion of Zeno behavior in \cite{fan2013distributed}  was pointed out in 
\cite{sun2018comments}. A clear-cut result on consensus of a general linear MAS was developed in 
\cite{hu2016consensus} where we believe an incorrect statement exists in the proof.  In particular, the proof 
confused $\lim_{k\to \infty} q_i(t^i_k) =0$ and $q_i(t^i_k) =0$ and thus a mistake followed.  An error also exists 
in the proof given in \cite{du2017distributed}.  The inter-event interval for each agent should be calculated based 
on its individual triggering function, that is, the change rate of $\|e_i(t)\|/\|z_i(t)\|$ using the notations in the paper. 
However, the change rate of $\|e(t)\|/\|z(t)\|$ of the whole network was used in the proof. 
Thus, the subsequent statement failed.
 
Next, let us discuss various approaches in achieving the two targets of  exclusion of Zeno behavior and asymptotic control performance,  at different additional costs. The cost of the first approach is using a synchronous clock mechanism. 
 For example, in \cite{liu2012event}, the control input for an agent is updated both at its own event time  as well as that of its neighbors, 
 which results in a synchronous triggering clock for all the agents.  
 In \cite{nowzari2016distributed}, during a certain time interval, 
 an agent receives new information from the broadcast of a neighbor and then immediately broadcasts its state, assuming 
 no time-delay in information broadcasting. It also results in a synchronous triggering clock 
for a set of connected agents.  A synchronous clock simplifies the triggering mechanism in a network
and guarantees the two targets.

The idea of inclusion of a constant threshold in event function has been discussed 
in  \cite{aastrom1999comparison, borgers2014event}.   An extension of this idea is to replace 
a constant threshold by a time-dependent function.  On one hand, this positive
time-dependent function sets the lower bound of inter-event intervals; on the other hand, 
it, decreasing to zero with time, drives the control error to zero. This is the second approach for achieving 
the aforementioned two targets; see, e.g.,   \cite{seyboth2013event,girard2015dynamic,sun2016new}. 
One cost of this approach is that the  time dependent function
makes the closed-loop system non-autonomous and the lower bound of the
inter-event intervals relies on time, or the initial states.  

The third approach is to combine event-triggered control and sampled-data control. 
For example, in \cite{meng2013event,guo2014distributed},  
the sampling period for all agents is synchronized by a clock and
the sampled data are used for event detection. 
Whether or not sampled data of agent should be used for actuation and broadcasted at the sampling instant 
depends on whenever an event is detected, which makes the protocol different from sampled-data control. 
In this setting, it is easy to prove that the inter-event interval is at least one sampling period, excluding Zeno behavior.  
The cost is again the existence of a synchronous  clock.

The fourth approach also takes the advantage of sampled-data control in naturally 
setting the lower bound of inter-event intervals as at least one sampling period. 
But this approach does not require a synchronous clock.
The idea is to trigger an event when a normal triggering condition is satisfied (which by itself 
may result in Zeno behavior) AND the inter-event interval is larger than a specified 
sampling period (called a fixed timer).  The cost is that the global information about 
the network is required in determining the sampling period. In our opinion, this cost 
is the minimal in the existing approaches. 
This approach was first used in \cite{mazo2011decentralized,tallapragada2014decentralized} 
for event-triggered stabilization problems.  It was also used for consensus of 
first-order integrators in \cite{fan2015self} for a network of an undirected and connected graph. 
The result was further generalized to general linear MASs in  
\cite{hu2017output}, also for an undirected and connected network;
and in \cite{xu2017event}  for a leader-following network containing an undirected subgraph of the follower network. 
It is noted that,  in \cite{xu2017event},  an agent updates the control protocol at its
own triggering instants and its neighbors' triggering instants as well.

In this paper, we apply the aforementioned fourth approach 
to the reference model consensus problem. The main contribution is to extend this 
most cost-efficient method for handling  
general linear MASs from an undirected network to a directed network. 
 It is challenging since the existing works for dealing with undirected topology rely on the diagonalization of 
the Laplacian matrix and they do not apply to the asymmetry associated with the Laplacian of a directed graph. 
Moreover, a clear-cut condition for triggering event in terms of 
an explicit mathematical expression is given in this paper, while 
it was hidden in complicated algorithmic descriptions in the existing works, e.g., \cite{hu2017output,xu2017event}.

%
%

\medskip
\noindent {\it  Problem 2- Event-Triggered Perturbed Output Regulation}


Compared with the bulky literature on event-triggered control of linear MASs, 
the research results on event-triggered control of nonlinear systems are much limited. 
One typical result on event-triggered stabilization for a nonlinear system 
is given in  \cite{liu2015small} using the
small gain theorem based on the prerequisite that the nonlinear system 
with a continuous-time (not event-triggered) controller 
can be input-to-state stable (ISS)  from the measurement error to state. 
The idea was also used for partial state feedback and output feedback scenarios in \cite{Liu2015}.

This event-triggered stabilization technique was also extended to deal with
the output regulation problem, integrated with the internal model technique. 
The output regulation problem, aiming at reference tracking and disturbance rejection,
is more challenging than stabilization, because it involves a dynamic
internal model design and requires a technique to perform the emulation of dynamic
compensator in an  event-triggered manner. 
For example, the event-triggered output regulation
problem was studied in \cite{liu2017event} 
for a class of nonlinear systems, but it can only confine the steady-state tracking error to a prescribed bound, 
when Zeno behavior is excluded.  
The technique was extended to handle the cooperative output regulation of nonlinear MASs in 
\cite{liu2017cooperative}, and later in \cite{wang2018event} for the output feedback design.  
Again, the problem in these papers was  solved practically in the sense that the  tracking error is bounded rather than asymptotically approaching zero.

There were two technical disadvantages in the aforementioned works on event-triggered stabilization and/or
output regulation.  The first one  is that 
the ISS property, from the measurement (sensor) error to state,
is technically difficult to achieve because 
global internal stabilizability does not imply global external
stabilizability for nonlinear systems, even for small sensor disturbances.
The second one is the steady-state error in handling the output regulation problem.
Very recently, these two disadvantages were lightened in  \cite{Khan2018}
where a new event-triggered stabilization technique was proposed that
requires a milder condition that the controlled system has a certain ISS property, but from
the  input (actuator) disturbance to state. Also, the technique was further applied to
design an asymptotic output regulation controller.
 The framework for event-triggered control in  \cite{Khan2018} was motivated by the sampled-data
technique developed in \cite{chen2014performance}.
 
In this paper, we aim to further extend the technique developed in \cite{Khan2018} for an individual nonlinear system
to the networked scenario. In dealing with a networked MAS, the reference model designed in the first step works as
an exosystem for the output regulation problem in the second step. Obviously, each reference model is non-autonomous by itself
as it is subject to network influence. Such an output regulation problem with a non-autonomous exosystem
is called a perturbed output regulation problem. In summary, another main contribution of this paper 
is to develop a new event-triggered controller for the perturbed output regulation problem. 

The rest of this paper is organized as follows. In Section~II,  the robust output synchronization problem  is formulated 
and divided into two problems in a rigorous framework.  The two problems are studied in Sections III and IV, respectively. 
A numerical example is presented in Section~V.
Finally, some conclusions are drawn in Section~VI.

%

\section{Problem Formulation}               
Consider $N$ nonlinear agents described by the following lower triangular equations:
\EQ
\dot{z}_i&=f_{i0}(z_i,x_{i1},w_i) \\
\dot{x}_{i1}&=f_{i1}(z_i,{x_{i1}},w_i)+b_{i1}(w_i)x_{i2}\\
&\vdots \\
\dot{x}_{ir_i}&=f_{ir_i}(z_i,{x_{i1}},\cdots,{x_{ir_i}},w_i)+b_{ir_i}(w_i)u_i\\
y_i&=x_{i1},\;\;\;i=1,\cdots,N, \label{agent}
\EN
where $z_i\in\mathbb{R}^{m_i}$ and $x_{i1},\cdots, x_{i r_i}\in\mathbb{R}$  are the agent states and $u_i\in\mathbb{R}$ is the agent input, and $y_i\in\mathbb{R}$ is the output.  The vector $w_i\in\mathbb{W}_i$ represents system uncertainties for a compact set $\mathbb{W}_i \subset\mathbb{R}^{l_i}$. The system has a relative degree $r_i$. The function $f_{is},$ $s = 0,\cdots,r_i$, satisfies $f_{is}(0,\cdots,0,w_i) = 0$ for all $w_i\in\mathbb{W}_i$, which ensures the equilibrium point at the origin. 

The main aim of this paper is to design a distributed controller $u_i$ for each agent in 
(\ref{agent})  in an event-triggered manner such that the agents altogether achieve output synchronization to an agreed trajectory $y_{\infty}(t)$ in the sense of 
\EQ
\lim_{t\to\infty}\|  y_i(t)-y_{\infty}(t)\| =0,\;\;i=1,\cdots,N.\label{syn}
\EN
We consider the scenario that $y_{\infty}(t)$ is neither prescribed a priori nor known to any agent, but it satisfies a specific pattern governed by
\EQ
 y_{\infty}(t)=c(v_{\infty}(t)) \\
 \lim_{t\rightarrow\infty} [\dot{v}_{\infty}(t)-Av_{\infty}(t)] =0 
. \label{pattern}
\EN
where $v_{\infty}\in\mathbb{R}^q$, $A$ is a fixed matrix, and $c(\cdot)$ is a sufficiently smooth function.

Suppose the communication network among the agents  in (\ref{agent}) is modeled through a directed graph, $\mathcal{G}$. For  a directed graph $\mathcal{G}=(\mathcal{V},\mathcal{E})$, the in-neighbor set $\mathcal{N}_i$ is defined as $\mathcal{N}_i=\{j\in \mathcal{V}\vert(j,i)\in\mathcal{E}\}$ and the out-neighbor set $\mathcal{M}_i$ is defined as $\mathcal{M}_i=\{j\in \mathcal{V}\vert(i,j)\in\mathcal{E}\}$. Let $a_{ij}$ denotes the weight for the edge $(j, i)\in\mathcal{E}$ and is set to positive weight if $(j, i)\in\mathcal{E}$ and $0$ if $(j, i)\notin\mathcal{E}$. The (asymmetric) Laplacian matrix $\mathcal{L} = [l_{ij} ] \in \mathbb{R}^{N\times N}$ associated with  $\mathcal{G}$ is defined as $l_{ii} =\sum_{j=1,j\neq i}^{N}a_{ij}$ and $l_{ij}=-a_{ij},i\neq j.$ Let ${\mathbf r}, \mathbf{1} \in \mathbb{R}^N$, with  $\mathbf{r}\t \mathbf{1}=1$, be the left and right eigenvectors of $\mathcal{L}$ associated with the zero eigenvalue, respectively. In particular, 
$\mathbf{1}$ is the vector with all elements being 1. The vector ${\mathbf r}$ has all positive elements, ${\mathbf r_1},\cdots,{\mathbf r}_N>0$, and is called a positive vector. 
Denote $R := {\rm diag}\{\mathbf r_1,\cdots,\mathbf r_N\}$.

For each agent, we create a linear reference model, with $v_i\in\mathbb{R}^q$, as follows,
\begin{eqnarray}
\dot{v}_i(t)=Av_i(t)+B{\mu}_i(t),\;\;i=1,\cdots,N. \label{exosys}
\end{eqnarray}
With proper design of the input $\mu_i\in\mathbb{R}$, it is expected that the reference models will achieve 
consensus on a trajectory $v_{\infty}(t)$, which does not 
physically exist though.  The matrix $A$ is the one specified in (\ref{pattern}) and $B\neq 0$ 
is chosen to satisfy Assumption~\ref{as:1} later.

With the above setting, the objective now is to formulate  suitable event-triggering mechanisms and the associated sampled-data controllers $\mu_i$ and $u_i$ in order to achieve the state consensus of 
reference models and output regulation of  multi-agents, respectively. The rigorous formulations for these two processes are given below after some notations are introduced.

The relative measurement of reference states is defined as 
\begin{eqnarray}
{p}_i(t):=\sum_{j\in\mathcal{N}_i}a_{ij}(v_j(t)-v_i(t)), \label{eqn4}
\end{eqnarray}
which will be used for consensus controller design. 
Denote \begin{eqnarray}
e_i(t):=y_i(t)-c(v_i(t)).\label{trackerror}
\end{eqnarray}
For $i=1,\cdots, N$, denote the sets $\mathbb{S}\c_i=\{0,1,\cdots, k_i\c\}$ for a finite integer $k_i\c$ or $\mathbb{S}\c_i=\mathbb{Z}_+$ where $\mathbb{Z}_+$
is the set of non-negative integers. 
Similarly, denote the sets $\mathbb{S}\r_i=\{0,1,\cdots, k_i\r\}$ for a finite integer $k_i\r$ or $\mathbb{S}\r_i=\mathbb{Z}_+$.

\medskip

\noindent {\textbf{Problem 1- Event-triggered Reference Model Consensus}}: 
For $i=1, \cdots, N$, 
design a distributed event-triggering mechanism for generating a time sequence $t\c_{ik},\; k \in \mathbb{S}\c_i$,
\footnote{
If $\mathbb{S}\c_i$  is a finite set, we let $t\c_{i,k\c_i+1} =\infty$ for complement of notation.
}  satisfying 
 $t\c_{i0}=0$ and $t\c_{ik}<t\c_{i,{k+1}}$, and a consensus controller 
\begin{eqnarray}
\mu_i(t)=g_iKp_i(t\c_{ik}), \;\;t\in[t\c_{ik},t\c_{i,k+1}) \label{eventcont0}
\end{eqnarray}
for some gains $K$ and $g_i$, 
such that the reference models (\ref{exosys}) achieve consensus in the sense of
\begin{eqnarray}
\lim_{t\to\infty}\|  v_{i}(t)-v_\infty (t)\| =0 \label{eqn3}
\end{eqnarray}
for an agreed trajectory $v_{\infty}(t)$  satisfying  (\ref{pattern}).  


\medskip

\noindent {\textbf{Problem 2- Event-triggered Perturbed  Output Regulation}}: 
For $i=1, \cdots, N$, 
design a distributed event-triggering mechanism for generating a time sequence $t\r_{ik},\; k \in \mathbb{S}\r_i$,  satisfying 
 $t\r_{i0}=0$ and $t\r_{ik}<t\r_{i,{k+1}}$,  
and an output regulation controller 
\begin{eqnarray}
\bar u_i(t)=\kappa_i(\bar x_i(t\r_{ik})), \;\;t\in[t\r_{ik},t\r_{i,k+1})
\label{regucont}
\end{eqnarray}
with dynamics sensor and actuator compensators (filters)
\EQ
u_i(t) &=  \beta^{\rm u}_i (\bar u_i(t), \eta_i(t)) \\
\bar x_i(t) &=  \beta^{\rm x}_i (x_i(t), \eta_i(t), e_i(t))  \\
\dot{\eta}_i(t)&=\varkappa_i(\eta_i(t), x_i(t), u_i(t)), \label{filter}
\EN
such that the agent (\ref{agent}) with the reference model (\ref{exosys}) achieves the property
\begin{eqnarray}
\lim_{t\to\infty}\|  e_{i[t,\infty)}\| \leq\hat{\gamma}_{i}\Big(\lim_{t\to\infty}\|  \mu_{i[t,\infty)}\| \Big) \label{reguerror1}
\end{eqnarray}
for some $\hat{\gamma}_i\in\mathcal{K}$. 


\begin{theorem}\label{synchronization}
Consider the MAS (\ref{agent}) and the reference model (\ref{exosys}). Suppose Problem 1 and Problem 2 are solved separately.  Then, the closed-loop system composed of (\ref{agent}), (\ref{exosys}) and the controllers 
(\ref{eventcont0}),   (\ref{regucont}),  (\ref{filter}) achieves 
output synchronization in the sense of (\ref{syn})
for an agreed trajectory $y_{\infty}(t)$ satisfying  (\ref{pattern}).  
\end{theorem}

\begin{proof}
When (\ref{eqn3}) is accomplished in Problem 1, for each reference model, one has
\EQQ
\lim_{t\to\infty}p_i(t) =\lim_{t\to\infty}\sum_{j\in\mathcal{N}_i}a_{ij}(v_j(t)-v_i(t)) =0.
\ENN
If $\mathbb{S}\c_i=\mathbb{Z}_+$, one has
\EQQ
\lim_{k \to\infty} p_i(t\c_{ik}) = \lim_{t\to\infty}p_i(t) =0
\ENN
and hence
\EQQ
\lim_{t\to\infty}\mu_i(t) =\lim_{k \to\infty} g_iKp_i(t\c_{ik}) = 0.\ENN
If $\mathbb{S}\c_i$  is a finite set, 
one has 
\EQQ
\mu_i(t)=g_iKp_i(t\c_{i {k_i\c}}), \;\;t \geq t\c_{i{k_i\c}},
\ENN
which is a constant. 
It is easy to see that the signal $\tilde v_i(t)  = v_{i}(t)-v_\infty (t)$ satisfies 
\EQQ
\dot {\tilde v}_i(t)
= A\tilde v_i(t)+B{\mu}_i(t) - [\dot{v}_{\infty}(t)-Av_{\infty}(t)], 
\ENN
which implies, noting $B\neq 0$,
\begin{eqnarray}
\mu_i(t)=0, \;\;t \geq t\c_{i{k_i\c}}.
\end{eqnarray}
Otherwise, it contradicts with  (\ref{eqn3}).

Therefore, in any case, one has 
$\lim_{t\to\infty}\mu_i(t) = 0$, which, together with 
(\ref{reguerror1}) in Problem 2, implies
$\lim_{t\to\infty}e_i(t) = 0$.
In other words, one has
\EQQ
\lim_{t\to\infty} \|y_i(t)-c(v_i(t))\|=0
\ENN
and hence 
\EQQ
\lim_{t\to\infty} \|y_i(t)-c(v_\infty(t))\|=0.
\ENN
Therefore, (\ref{syn}) is verified.  
\end{proof}

\begin{figure}[hbtp]
\centering
\includegraphics[scale=0.42]{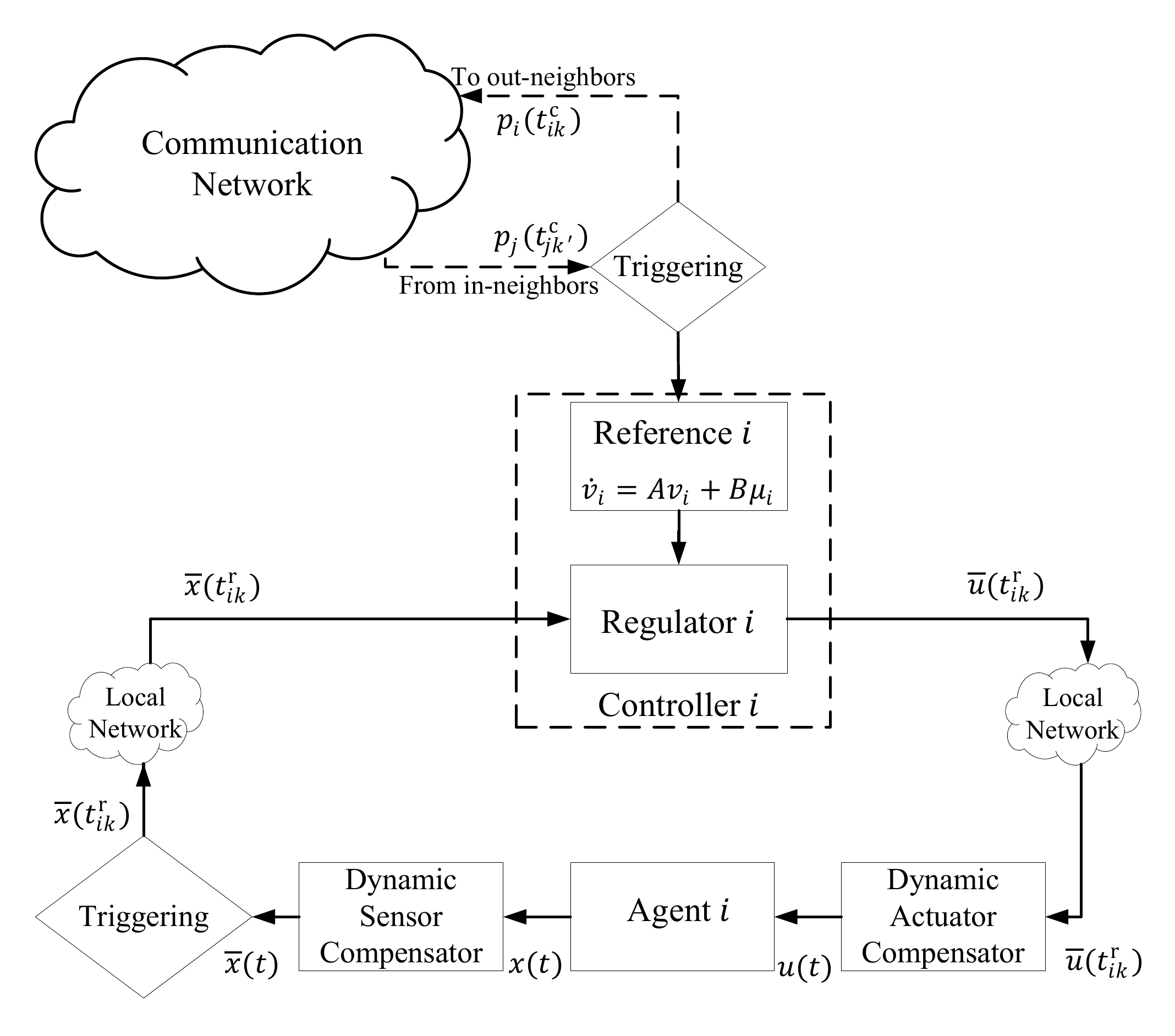}
\caption{Schematic diagram of event-triggered control design for each individual agent.\label{fig:1}}
\end{figure}

\begin{remark}
The schematic diagram of event-triggered control design for each individual agent is depicted in Fig.~\ref{fig:1}.  The overall synchronization process is broken down into two steps or processes. The event-triggered output regulation controller for each individual agent is actually decoupled from the network and hence completely distributed. This means that the event-triggering mechanisms for both linear consensus and nonlinear output regulation for each agent are designed independently. In the first step, once an event is triggered at $t_{ik}\c$, $k\in\mathbb{S}_i\c$, for agent $i$, it executes the consensus control actuation and broadcasts the state $p_i(t_{ik}\c)$ to 
its neighbors. Meanwhile, it also receives the information broadcasted from its neighbors, say $p_j(t_{jk'}\c)$ from agent $j$, at  
$t_{jk'}^c$, $k'\in\mathbb{S}_j\c$.  It is worth noting that, receiving $p_j(t_{jk'}\c)$ does not trigger an event for agent $i$, 
but it will be used to determine its event-triggering condition.   In the second step, an event is triggered at $t_{ik}\r$, $k\in\mathbb{S}_i\r$ for each agent and the measurement is transmitted via a local network to the regulation controller that is immediately executed. No broadcast among agents occurs in this step.  For each agent, the triggering instants $t_{ik}\c$ and $t_{ik}\r$ are completely independent. 
\end{remark}

Theorem~\ref{synchronization} shows that the solvability of the output synchronization problem relies on 
that of Problem~1 and Problem~2. The two problems will be investigated in the subsequent two sections, respectively.
Moreover, it is an important issue in event-triggered control 
that the so-called Zeno behavior must be excluded in the sense
of
\EQ \inf_{k\in\mathbb{S}\c_i}\{t\c_{i,k+1}-t\c_{ik}\}>0,\;
\inf_{k\in\mathbb{S}\r_i}\{t\r_{i,k+1}-t\r_{ik}\}>0, \\ \; i=1,\cdots,N. \label{zeno}
\EN
Obviously, with (\ref{zeno}), 
an infinite number of events never occur in a finite time period.
This issue will also be completely addressed in the following two sections.

\section{Problem 1- Event-triggered Reference Model Consensus}\label{consensus}

In this section, we will give an explicit solution to Problem~1, including an 
explicit mechanism for generating the time sequence $t\c_{ik},\; k \in \mathbb{S}\c_i$
and an explicit controller (\ref{eventcont0}).
For the convenience of analysis, we define the following signals
\EQ
p_i\c(t)&=p_i(t\c_{ik}), \;\;t\in[t\c_{ik},t\c_{i,k+1})  \\
\epsilon_i(t)&=p_i\c(t)-p_i(t),\;i=1,\cdots,N.\label{eqn5}
\EN
The signal $p_i\c(t)$ updated at the triggering time $t\c_{ik}$ remains unchanged until the next triggering instant $t\c_{i,k+1}$; $\epsilon_i(t)$ represents the error between $p_i\c(t)$ and $p_i(t)$.
With $K=B\t P$ for a matrix $P$ to be determined later, the controller (\ref{eventcont0}) can be rewritten as 
follows, for $i=1,\cdots,N$,
\begin{eqnarray}
\mu_i(t)=g_i B\t P p_i\c(t). \label{eventcont}
\end{eqnarray}

Then, the closed-loop system comprising of (\ref{exosys}) and (\ref{eventcont}) can be written as
\begin{eqnarray}
\dot{v}_i(t)=Av_i(t)+g_iBB\t P(\epsilon_i(t)+p_i(t)),\; i=1,\cdots,N.\label{cls}
\end{eqnarray}
Define the following lumped symbols
\EQQ
v(t)&=\mbox{col}(v_1(t),\cdots,v_N(t))\\ 
\epsilon(t)&=\mbox{col}(\epsilon_1(t),\cdots,\epsilon_N(t))\\
G&={\rm diag}(g_1,\cdots,g_N) \\
p&=\mbox{col}(p_1,\cdots,p_N).
\ENN
The equations in (\ref{cls}) can be put in a compact form
\begin{eqnarray}
\dot{v}=(I_N\otimes A-\mathcal{L}G\otimes BB\t P)v +(G\otimes BB\t P)\epsilon \label{vlump}
\end{eqnarray}
where $I$ denotes an  identity matrix whose dimension is specified 
by its subscript. 
 With $p=-(\mathcal{L}\otimes I_q)v$, (\ref{vlump})  
can be further rewritten as   
\begin{eqnarray}
\dot{p}=(I_N\otimes A-\mathcal{L}G\otimes BB\t P)p -(\mathcal{L}G\otimes BB\t P)\epsilon .\label{eqn9}
\end{eqnarray}

We introduce two assumptions and two technical lemmas before 
the main statement. 

\begin{assumption}\label{as:1}
All the eigenvalues of $A$ are simple and have zero real parts and the pair $(A, B)$ is stabilizable.
\end{assumption}
\begin{assumption}\label{as:2}
The directed graph $\mathcal{G}$ is strongly connected, that is, there exists  a directed path from every node
to every other node.
\end{assumption}

\begin{lemma}\cite{mei2014consensus}\label{L2}
Under Assumption~\ref{as:2}, $\hat{\mathcal L} = R \mathcal{L}+\mathcal{L}\t R $ is a symmetric Laplacian  matrix associated with an undirected graph.  Let $\varsigma\in\mathbb{R}^N$ be any positive vector. Then, the following inequality holds
\begin{eqnarray}
\min_{{\upsilon}\t \varsigma=0,\upsilon\neq 0 }\frac{\upsilon\t \hat{\mathcal L}\upsilon}{\upsilon\t \upsilon}>\frac{\lambda_2(\hat{\mathcal L})}{N} \label{eqn10}
\end{eqnarray}
where $\lambda_2(\hat{\mathcal L})>0$
is the second-smallest eigenvalue of   $\hat{\mathcal L}$.
\end{lemma}

The next lemma gives an event-triggering mechanism under which 
the system holds a critical property that will be used in the main theorem. 

 \begin{lemma} \label{lem:st}
Consider the reference model (\ref{exosys})   with 
the controller (\ref{eventcont}) and
a time sequence $t\c_{ik},\; k \in \mathbb{S}\c_i$.
For $\eta_i >0 ,\; i=1,\cdots,N$,  let \EQQ
s_{ik}=\frac{\eta_i}{1+\eta_i}\| p_i(t\c_{ik}) \|
\ENN 
and  
\EQQ w_{ik}=&\| (A   -g_iBB\t P\sum_{j\in\mathcal{N}_i}a_{ij} ) p_i(t\c_{ik}) \| \\
w_i(t) =& \|BB\t P\sum_{j\in\mathcal{N}_i}a_{ij}g_jp_j\c(t) \|.
\ENN
For \begin{eqnarray}
\tau_{ik}=\inf\{t>0: \int_{t\c_{ik}}^{t\c_{ik} +t} [\| A\| s_{ik} + w_{ik} + w_{i}(\tau) ]d\tau =s_{ik} \}, \label{trig3}
\end{eqnarray}
one has 
\EQ
 \|\epsilon_i(t)\| \leq \eta_i\| p_i(t)\|,\; \forall t \in [t\c_{ik}, t\c_{ik} +\tau_{ik}). \label{epsilonipi}
 \EN
  \end{lemma}
  
\begin{proof} We first calculate the change rate of $\|\epsilon_i(t)\|$ as follows
\EQQ
\frac{d}{dt}\|  \epsilon_i(t)\|  \leq&\frac{\|\epsilon_i\t \|}{\|\epsilon_i\|}\|\dot{\epsilon_i}\|=\|\sum_{j\in\mathcal{N}_i}a_{ij}(\dot{v}_j(t)-\dot{v}_i(t))\|\nonumber\\
\leq&\|A \sum_{j\in\mathcal{N}_i}a_{ij}({v}_j(t)-{v}_i(t)) \\
&+\sum_{j\in\mathcal{N}_i}a_{ij}B(g_j B\t P p_j\c(t)-g_i B\t P p_i\c(t))\|
\nonumber\\
\leq&\|A p\c_i(t) -A\epsilon_i(t) \\ & +\sum_{j\in\mathcal{N}_i}a_{ij}B(g_j B\t P p_j\c(t)-g_i B\t P p_i\c(t))\|
\nonumber\\
\leq& \|  A\| \|  \epsilon_i(t)\|  +w_{ik} + w_i(t).
\ENN
Since $\|\epsilon_i(t\c_{ik})\|=0$,  for $t \in [t\c_{ik}, t\c_{ik} +\tau_{ik})$,
\EQQ
 \|\epsilon_i(t)\|  = \|\epsilon_i(t\c_{ik})\| +
 \int_{t\c_{ik}}^{t\c_{ik} +t}  \frac{d}{d\tau }\|  \epsilon_i(\tau)\|  d\tau \\
 \leq  \int_{t\c_{ik}}^{t\c_{ik} +t}     [ \|  A\| \|  \epsilon_i(\tau)\|  +w_{ik} + w_{i}(\tau) ]d\tau 
  \leq s_{ik},
\ENN
and 
\EQQ
 \|\epsilon_i(t)\|  \leq \frac{\eta_i}{1+\eta_i}\| p_i(t\c_{ik})\|
=\frac{\eta_i}{1+\eta_i}\| p_i(t) +\epsilon_i(t) \|\\
\leq \frac{\eta_i}{1+\eta_i}\| p_i(t)\|  +\frac{\eta_i}{1+\eta_i} \|\epsilon_i(t) \|.
\ENN
As a result, one has
\EQQ
 \|\epsilon_i(t)\|   \leq \eta_i \| p_i(t)\|.  \ENN
The proof is thus completed. 
\end{proof}

\begin{remark}
To calculate $\tau_{ik}$, the agent $i$ needs $s_{ik}$, $w_{ik}$, and $w_i(t)$.
It is easy to see that $s_{ik}$ and $w_{ik}$ depend on $p_i(t\c_{ik})$
that is measured at the triggered time $t\c_{ik}$. The signal 
$w_i(t)$ depends $p_j\c(t)=p_j(t\c_{jk^\prime})$ from its neighbors $j\in\mathcal{N}_i$, 
which is measured at the triggered time $t\c_{jk^\prime}$.  In other words, when one agent 
has an event to update its $p_j\c(t)$, it broadcasts the information to its neighbors. 
Using this mechanism, no agent needs to monitor the network continuously. 
\end{remark}
 
Now, the main result is stated in the following theorem.  

\begin{theorem}\label{T1:2} 
Pick $g_i \geq {\mathbf r}_i >0$, $i=1,\cdots, N$.
Under Assumptions~\ref{as:1}-\ref{as:2}, there always exists a unique solution 
$P > 0$ to the following Algebraic Riccati Equation:
\begin{eqnarray}
PA+A\t P-\lambda PBB\t P+\beta I=0\label{eqnr}
\end{eqnarray}
where 
\EQQ 0<\lambda<\frac{\lambda_2(\hat{\mathcal L})}{N},  
\beta=\frac{1}{\lambda_{\min}(GR )}.
\ENN
Let 
\EQQ
\rho=& \frac{\|  R \mathcal{L}G\otimes B\t P\|  }{\sqrt{{\lambda_2(\hat{\mathcal L})}/{N} - \lambda}} \\
b_1=&\|A\|+\lambda_{\mathcal{L}G}\|BB\t P\| \\
b_2=&\lambda_{\mathcal{L}G}\|BB\t P\|
\ENN
where  $\lambda_{\mathcal{L}G}$ denotes the largest eigenvalue of ${\mathcal{L}G}$.
Pick $\eta\geq \eta_i >0 ,i=1,\cdots,N$, and  $\phi>0$ satisfying  
\EQ
\varphi =  \rho^2  \eta^2    + N \rho^2 \phi^2  <1. \label{varsigma}
   \EN
Consider the reference model (\ref{exosys})  with the control law 
(\ref{eventcont})
and the triggering mechanism  
\begin{eqnarray}
t\c_{i,k+1}=t\c_{ik}+\max\{\tau_{ik},b\}\label{maxtrigger}
\end{eqnarray}where $\tau_{ik}$ is determined in Lemma~\ref{lem:st} and
\begin{eqnarray}
b=\frac{\ln(\phi+1)}{b_1+b_2\max\{\eta, \phi\} \sqrt{N}}.\label{lowerbound}
\end{eqnarray} 
Then, Problem 1 is solved in the sense of (\ref{eqn3}).
\end{theorem}

\begin{proof}
Choose a candidate Lyapunov function  as follows
\begin{eqnarray}
V(p)=\frac{1}{2}p\t (GR \otimes P)p.\label{eqn10b}
\end{eqnarray}
 Along the trajectory of (\ref{eqn9}), the time derivative of $V(p)$ is obtained as
\begin{align}
\dot{V}(p)=&\frac{1}{2}p\t [GR \otimes (PA+A\t P)-G\hat{\mathcal L}G\otimes PBB\t P]p\nonumber\\
&-p\t (GR \mathcal{L}G\otimes PBB\t P)\epsilon.\label{eqdV}
\end{align}

Define a positive vector 
\EQQ \varsigma =G^{-1}{\mathbf r}    \in{\mathbb R}^N
\ENN 
and hence  
 \EQQ \upsilon=(G\otimes B\t P)p \in{\mathbb R}^N.
 \ENN
Since
\EQQ
\upsilon\t \varsigma&={p}\t (G\otimes PB) G^{-1}{\mathbf r}    \nonumber\\
&=-v\t (\mathcal{L}\t \otimes I_q)(\mathbf r\otimes  PB)\nonumber\\
&=-v\t (\mathcal{L}\t {\mathbf r} \otimes PB)=0, 
\ENN
one has, by using Lemma \ref{L2}, 
\EQ
 p\t (G\hat{\mathcal L}G\otimes PBB\t P) p  =
  p\t (G\otimes PB) \hat{\mathcal L} (G\otimes B\t P) p \\
   =\upsilon\t  \hat{\mathcal L} \upsilon \geq \frac{\lambda_2(\hat{\mathcal L})}{N}\upsilon\t   \upsilon =\frac{\lambda_2(\hat{\mathcal L})}{N}p\t (G^2\otimes PBB\t P)p. \label{eqdV1}
\EN


Next, denote  
\EQQ 
a = 2 {\lambda_2(\hat{\mathcal L})}/{N} -2 \lambda >0.
\ENN
So, \EQQ
\rho^2=\frac{2\|  R \mathcal{L}G\otimes B\t P\| ^2 }{ a}.
\ENN
 A direct calculation shows
\EQ
\; & -p\t (GR \mathcal{L}G\otimes PBB\t P)\epsilon  \\
= & -p\t (\sqrt{\frac{a}{2}}G\otimes PB)(\sqrt{\frac{2}{a}} R \mathcal{L}G\otimes B\t P)\epsilon \\
\leq & \frac{a}{4}p\t (G^2\otimes PBB\t P)p(t)+\frac{\rho^2}{2}\|\epsilon\|^2 .\label{eqdV2}
\EN
The condition  $g_i \geq \mathbf r_i$ yields that 
\EQQ G^2 \otimes PBB\t P \geq GR\otimes PBB\t P.
\ENN
Therefore, by substituting (\ref{eqdV1}) and (\ref{eqdV2}) into (\ref{eqdV}) 
and noting ${\lambda_2(\hat{\mathcal L})}/{N}-{a}/{2} =\lambda$, one has
\EQQ
\dot{V}(p)&\leq\frac{1}{2}p\t \Big[GR \otimes \Big(PA+A\t P-\lambda PBB\t P\Big)\Big]p+\frac{\rho^2}{2}
\|\epsilon\|^2\\
&\leq-\frac{1}{2}\beta\lambda_{\min}(GR ) \|p\|^2+\frac{\rho^2}{2}
\|\epsilon\|^2\\
&=-\frac{1}{2} \|p\|^2+\frac{\rho^2}{2}
\|\epsilon\|^2.
\ENN
If the following condition is always satisfied 
\begin{eqnarray}
\|\epsilon\|^2\leq\frac{\varphi}{\rho^2}\|p\|^2\label{combined condition}
\end{eqnarray}
with $\varphi <1$, one has
\begin{eqnarray}
\dot{V}(p) \leq - \frac{(1-\varphi)}{2}\| p\|^2. \label{error}
\end{eqnarray}
As a result, one can conclude that  
\EQ
\lim_{t\rightarrow\infty}p(t) &=-\lim_{t\rightarrow\infty} (\mathcal{L}\otimes I_q)v(t)=0 \\
\lim_{t\rightarrow\infty}\epsilon(t) &=0.
\label{asymptooticconsensus}
\EN
Under Assumption~\ref{as:2}, the Laplacian $\mathcal{L}$
has one zero eigenvalue and the other eigenvalues have positive
real parts. In particular,  there exist matrices $W\in\mathbb{R}^{\left(N-1\right)\times N}$,
$U\in\mathbb{R}^{N\times \left(N-1\right)}$ such that
 \EQQ
T =  \left[ \begin{array}{c} {\mathbf r}\t \\ W \end{array} \right] ,\;
T^{-1} =  \left[\begin{array}{cc}
{\bf 1}& U \end{array}\right].
\ENN
One has the following similarity transformation
\EQQ
T\mathcal{L}T^{-1}=\left[\begin{array}{cc}
0 & 0\\
0 & \mathcal{J}
\end{array}\right]
\ENN
where $\mathcal{J}=W\mathcal{L}U$
is a matrix with all eigenvalues having positive real parts.
From the definition of $T$ and $T^{-1}$, one has
\EQQ {\bf 1} {\mathbf r}\t + U W =I
\ENN
and
\EQQ
 W = \mathcal{J}^{-1}W\mathcal{L} UW =\mathcal{J}^{-1}W\mathcal{L} (I-{\bf 1} {\mathbf r}\t) =\mathcal{J}^{-1}W\mathcal{L}.
\ENN
which implies
\EQQ
\lim_{t\rightarrow\infty} (W \otimes I_q)v(t)=
\lim_{t\rightarrow\infty} (\mathcal{J}^{-1}W\otimes I_q)  (\mathcal{L} \otimes I_q)v(t)=0.
\ENN
Therefore, one has 
\EQQ v& = ({\bf 1} {\mathbf r}\t )\otimes I_q v+ (U W)\otimes I_q v  \\
&= {\bf 1} \otimes  ({\mathbf r}\t \otimes I_q) v+ (U\otimes I_q ) (W\otimes I_q) v\\
&= {\bf 1} \otimes  v_\infty + (U\otimes I_q ) (W\otimes I_q) v
\ENN
for  $v_\infty = ({\mathbf r}\t \otimes I_q){v}$. 
On one hand, one has
\EQQ
({\mathbf r}\t \otimes I_q) \dot{v}&= {\mathbf r}\t \otimes  A v +({\mathbf r}\t \otimes I_q) (G\otimes BB\t P)\epsilon  \\
&=A  ({\mathbf r}\t \otimes  I_q) v +({\mathbf r}\t \otimes I_q) (G\otimes BB\t P)\epsilon, \ENN
that is
\EQQ
 \dot{v}_\infty =A  v_\infty +({\mathbf r}\t \otimes I_q) (G\otimes BB\t P)\epsilon 
\ENN
which satisfies (\ref{pattern}). On the other hand, 
\EQQ \lim_{t\to\infty} \| v(t) -  {\bf 1} \otimes  v_\infty\| =
\lim_{t\to\infty}  \|(U\otimes I_q ) (W\otimes I_q) v(t)\| =0
\ENN
which is equivalent to   (\ref{eqn3}). Problem 1 is thus solved. 

\medskip

What is left is to verify (\ref{combined condition}).  In particular, we will first prove 
that 
\EQ
\|\epsilon_i(t)\|\leq \max\{ \eta_i\|p_i(t)\|,\phi\|p(t)\| \},\; \forall t\in [t\c_{ik}, t\c_{i,k+1}) 
\label{pf-epsilon}
\EN
in two different cases based on (\ref{maxtrigger}).

Case 1 ($t\c_{i,k+1} =t\c_{ik}+\tau_{ik}$): By Lemma~\ref{lem:st}, one has (\ref{epsilonipi})
which implies (\ref{pf-epsilon}).

Case 2 ($t\c_{i,k+1} =t\c_{ik}+b$): It is noted that $\|\epsilon_i(t)\|/\|p(t)\| = 0$ for $t=t\c_{ik}$. 
It suffices to show that it takes at least a time $b$ for $\|\epsilon_i(t)\|/\|p(t)\|$ to evolve from $0$
to $\phi$. For this purpose, we calculate the change rate as follows
\EQQ
\frac{d}{dt}\frac{\|\epsilon_i\|}{\|p\|}=&\frac{\epsilon_i\t \dot{\epsilon}_i}{\|\epsilon_i\|\|p\|}-\frac{\|\epsilon_i\|p\t \dot{p}}{\|p\|^3}\nonumber\\
\leq&\frac{\|\dot{\epsilon}_i\|}{\|p\|}+\frac{\|\epsilon_i\|\|\dot{p}\|}{\|p\|^2}\leq\frac{\|\dot{p}\|}{\|p\|}+\frac{\|\epsilon_i\|\|\dot{p}\|}{\|p\|^2}\nonumber\\
\leq&\frac{\|\dot{p}\|}{\|p\|}\Big(1+\frac{\|\epsilon_i\|}{\|p\|}\Big)\nonumber\\
\leq&\Big(b_1+b_2\frac{\|\epsilon\|}{\|p\|}\Big)\Big(1+\frac{\|\epsilon_i\|}{\|p\|}\Big).
\ENN

For $\|\epsilon_i(t)\|/\|p(t)\|$ not reaching $\phi$, one has 
\EQQ
\frac{\|\epsilon\|}{\|p\|}=\sqrt{\sum_{i=1}^N\Big(\frac{\|\epsilon_i\|}{\|p\|}\Big)^2}\leq\max\{\eta, \phi\} \sqrt{N}.
\ENN  
As a result,  
\EQQ
\frac{d}{dt}\frac{\|\epsilon_i\|}{\|p\|}\leq\Big(b_1+b_2 \max\{\eta, \phi\} \sqrt{N}\Big)\Big(1+\frac{\|\epsilon_i\|}{\|p\|}\Big).
\ENN
So, the time for $\|\epsilon_i(t)\|/\|p(t)\|$ change from $0$
to $\phi$ is at least  $b$ defined in (\ref{lowerbound}).

From above, (\ref{pf-epsilon}) is proved, which implies (\ref{combined condition}) since \EQ
\|\epsilon(t)\|^2 \leq  \eta^2  \|p(t)\|^2  + N \phi^2 \|p(t)\|^2 =
\frac{\varphi}{\rho^2}  \|p(t)\|^2 \EN
for $\varphi$ defined in (\ref{varsigma}). 
\end{proof}

\begin{remark}
With $\bar B=\sqrt{\lambda}B$ and $\bar C=\sqrt{\beta}I_n$, the equation (\ref{eqnr}) can be rewritten as
\begin{eqnarray*}
PA+A^TP-P\bar B \bar B\t P+\bar C\t \bar C=0,
\end{eqnarray*}
which is the  Algebraic Riccati Equation (ARE) given in \cite{kucera1972contribution}. Therefore, Assumption \ref{as:1}  guarantees the existence of $P$.
In Theorem~\ref{T1:2},  Zeno behavior is obviously excluded because 
\begin{eqnarray}
\inf_{k\in\mathbb{S}\c_i}\{t\c_{i,k+1}-t\c_{ik}\}\geq b,\; i=1,\cdots,N. \end{eqnarray}
\end{remark}
%

\section{Problem 2- Event-Triggered Perturbed Output Regulation}\label{regulation} 

In this section, we will move to an explicit solution to Problem~2. 
It is worth mentioning that the output regulation problem for each agent is disengaged from the network and hence the controller can be designed in a fully distributed fashion. 
For the ease of presentation, we simply omit the agent label $i$ of (\ref{agent}) and (\ref{trackerror}) throughout this section, 
which results in the following set of equations 
\EQ
\dot{z}&=f_0(z,x_1,w)\\
\dot{x}_1&=f_1(z,{x_1},w)+b_1(w)x_{2}\\
&\vdots  \\
\dot{x}_r&=f_r(z,{x_1},\cdots,{x_r},w)+b_r(w)u\\
e&=x_1-c(v)\\
\dot{v}&=Av+B{\mu}.
\label{eqnor1}
\EN
To solve Problem~2, we need the following standard assumptions as in \cite{chen2015stabilization}. 
\begin{assumption}\label{RA1}
For $j=1,\cdots,r$, $b_j(w)>0, \forall  w\in\mathbb{W}$.
\end{assumption}    

\begin{assumption}\label{RA2}
There exists sufficiently smooth function $\mathbf{z}(v,w),\mathbf{x}_1(v,w),\cdots,\mathbf{x}_{r+1}(v,w)$ satisfying the following regulator equations, for all $ v\in\mathbb{R}^q$ and $w\in \mathbb{R}^l$:
\EQ
\frac{\partial \mathbf{z}(v,w)}{\partial v}Av=&f_0(\mathbf{z}(v,w),c(v),w)\\
\mathbf{x}_1(v,w)=&c(v)\\
\frac{\partial \mathbf{x}_j(v,w)}{\partial v}Av=&f_j(\mathbf{z}(v,w),\mathbf{x}_1(v,w),\cdots,\mathbf{x}_j(v,w),w)\\
&+b_j(w)\mathbf{x}_{j+1}(v,w),\;\;j=1,\cdots,r.\label{eqnor5}
\EN
\end{assumption}

\begin{assumption}\label{RA3}
There exists a linear observable steady-state generator of the following form, with $\vartheta_j\in\mathbb{R}^{\ell_j}$
\begin{eqnarray}
\frac{\partial \vartheta_j(v,w)}{\partial v}Av=\Phi_j\vartheta_j(v,w),\;\mathbf{x}_{j+1}(v,w)=\Psi_j\vartheta_j(v,w)\label{eqnor6}
\end{eqnarray}
for an observable pair $({\Psi}_j, {\Phi}_j)$.
\end{assumption}
\noindent Under the Assumptions \ref{RA1}-\ref{RA3}, for some controllable pair of matrices $(M_j,N_j)$ with $M_j\in\mathbb{R}^{\ell_j\times \ell_j}$ Hurwitz, there exists a nonsingular matrix $T_j$ satisfying the Sylvester equation $M_jT_j+N_j\Psi_j=T_j\Phi_j$. Then, on mapping $\theta_j(v,w)=T_j\vartheta_j(v,w)$, we can construct an alternative steady-state generator as follows
\EQ
\frac{\partial \theta_j(v,w)}{\partial v}Av=&T_j\Phi_jT^{-1}_j\theta_j(v,w)\\
\mathbf{x}_{j+1}(v,w)=&\Psi_jT^{-1}_j\theta_j(v,w).\label{ssg}
\EN
In the event-triggered control scenario, we introduce two additional components, namely, a dynamic actuator compensator  
\EQ
u&=\bar{u}+\Psi_r T^{-1}_r\eta_r\\
\dot{\eta}_r&=M_r\eta_r+N_ru\label{dync}
\EN
and a dynamic sensor compensator
\begin{eqnarray}
\bar{x}_1&=&e\nonumber\\
\bar{x}_j&=&x_j-\Psi_{j-1}T^{-1}_{j-1}\eta_{j-1}, \;j=2,\cdots,r\nonumber\\
\dot{\eta}_j&=&M_j\eta_j+N_j{x}_{j+1},\; j=1,\cdots,r-1.\label{dyns}
\end{eqnarray}
The structure of these two components is primarily based on the internal model design  and corresponds with the steady-state generator defined in (\ref{ssg}). The necessity of these two additional components arises in the event-triggered output regulation problem when $v(t)$ is a time-varying signal. The details are explained in \cite{Khan2018} with the schematic diagram illustrated in Fig. \ref{fig:1}.

Perform the following coordinate transformation
\EQ
{z}_0&=z-\mathbf{z}(v,w)\\
z_j&=\eta_j-\theta_j(v,w)-b^{-1}_j(w)N_j\bar{x}_j,\;j=1,\cdots,r.\label{coordinate}
\EN
As a result, the system composed of (\ref{eqnor1}), (\ref{dync}) and (\ref{dyns}) is transformed into following form
\EQ
\dot{z}_0=&\bar{f}_0(z_0,\bar{x}_1,\nu)+\tilde{p}_0(\nu)\mu \\
\dot{z}_1=&p_1(z_0,z_1,\bar{x}_1,\nu)+\tilde{p}_1(\nu)\mu \\
\dot{z}_j=&p_j(z_0,z_1,\cdots,z_j,\bar{x}_1,\cdots,\bar{x}_j,\nu)+\tilde{p}_j(\nu)\mu \\
\dot{\bar{x}}_1=&\bar{f}_1(z_0,{z}_1,\bar{x}_1,\nu)+b_1(w)\bar{x}_2+\tilde{f}_1(\nu)\mu \\
\dot{\bar{x}}_j=&\bar{f}_j(z_0,z_1,\cdots,z_j,\bar{x}_1,\cdots,\bar{x}_j,\nu)+b_j(w)\bar{x}_{j+1}, \\
& j=2,\cdots,r\label{compactsys}
\EN
where $\bar{x}_{r+1}=\bar{u}$, $\nu={\rm col}(v,w)$ and all the functions are properly defined. Let $\bar{x}=[\bar{x}_1,\cdots,\bar{x}_r]\t$ and $\xi=[z_0\t,\cdots,z_r\t,\bar{x}\t]\t.$  
The system (\ref{compactsys}) can be rewritten in a compact form as
\EQ
\dot{\xi}=f(\xi,\nu)+E(w)\bar{u}+\bar{\mu}\label{newsys}
\EN
with $\bar{\mu}=[\tilde{p}_0(\nu),\cdots, \tilde{p}_r(\nu),\tilde{f}_1(\nu),0,\cdots,0]\t \mu$, $E(w)=[0,\cdots,0,b_r(w)]\t$ and $f(0,\nu)=0$. Let the event-triggered controller $\bar{u}(t_k)$ be of the following form
\begin{eqnarray}
\bar{u}(t)=\kappa(\bar{x}(t\r_k)),\;t\in[t\r_k,t\r_{k+1}).\label{controller}
\end{eqnarray}
Denote
\begin{eqnarray}
\varpi(t)=\kappa(\bar{x}(t\r_k))-\kappa(\bar{x}(t)),\;t\in[t_k,t_{k+1}).\label{newsignals}
\end{eqnarray}
Then, the system (\ref{newsys}), with additional output $q(t)$, can be rewritten as follows
\EQ
\dot{\xi}(t)=&f_c(\xi(t),\nu(t))+E(w)\varpi(t)+\bar{\mu} \\
q(t)=&\frac{d\kappa(\bar{x}(t))}{dt}\label{newsys_pertub}
\EN
where
\begin{eqnarray*}
f_c(\xi(t),\nu(t))=f(\xi(t),\nu(t))+E(w)\kappa(\bar{x}(t)).
\end{eqnarray*}

The following minimum-phase assumption is required for the output regulation problem. 

\begin{assumption} \cite{chen2015stabilization}\label{RA4}
There exists a quadratic ISS Lyapunov function $V_0(z_0)=z_0\t S_0z_0$ whose derivative, along the trajectory $\dot{z_0}=\bar{f}_0(z_0,\bar{x}_1,\nu)$, satisfies
\begin{eqnarray*}
\dot{V}_0(z_0)\leq-\alpha_0(\|  z_0\| )+\sigma_0(\|  \bar{x}_1\| )
\end{eqnarray*}
for $\alpha_0\in\mathcal{K}_{\infty}$  and $\sigma_0\in\mathcal{K}$. Also
\begin{eqnarray}
\lim_{s\to 0^{+}} \sup\frac{s^2}{\alpha_0(s)}<+\infty,\;\lim_{s\to 0^{+}} \sup\frac{\sigma_0(s)}{s^2}<+\infty.\label{eqnor22}
\end{eqnarray}
\end{assumption} 

Throughout the paper, the symbols  $\cal K$, $\mathcal{K}_{\infty}$, and $\cal KL$ represent 
the sets of class $\cal K$, class $\mathcal{K}_{\infty}$, and class $\cal KL$ functions, respectively.

Now, the main result is summarized in the following theorem.
For the system (\ref{newsys_pertub}) with the ideal case $\bar\mu=0$, the same theorem was given in \cite{Khan2018}.
Due to this additional perturbation, the problem studied in this paper, called a 
perturbed output regulation problem, becomes more complicated. 
 
\begin{theorem}\label{T1}
Consider the system (\ref{eqnor1})  with the 
dynamic actuator compensator  (\ref{dync}), 
the dynamic sensor compensator (\ref{dyns}), and the event-triggered controller (\ref{controller}), 
under Assumptions~\ref{RA1}-\ref{RA4}. 
Suppose $\mu$ and $v$ are bounded. There exists a sufficiently smooth function $\kappa$ in (\ref{controller}) such that
the system (\ref{newsys_pertub}) satisfies the following input-to-state stability (ISS) and input-to-output stability (IOS) properties:
\EQ
\|  \xi(t)\| \leq&{\rm max}\{\tilde{\beta}(\|  \xi(\tau_0)\| ,t-\tau_0),\tilde{\gamma}(\| \varpi_{[\tau_0,t]}\| ), 
\\ & \tilde{\gamma}_{\bar{\mu}}(\| \bar{\mu}_{[\tau_0,t]}\| )\} ,\; \forall t\geq\tau_0\geq 0 \\ \label{eqn16a}
\|  q(t)\| \leq&{\rm max}\{\beta(\|  \xi(\tau_0)\| ,t-\tau_0),{\gamma}(\|  \varpi_{[\tau_0,t]}\| ), \\
& {\gamma}_{\bar{\mu}}(\| \bar{\mu}_{[\tau_0,t]}\| )\},\; \forall t\geq\tau_0\geq 0 
\EN
for $\tilde{\beta}, \beta\in\cal{KL}$ and $\tilde\gamma , \gamma,\tilde{\gamma}_{\bar{\mu}}$ and ${\gamma}_{\bar{\mu}}\in\cal{K}_\infty$. 
The function $\gamma$ is locally Lipschitz.
Let the event-triggering mechanism be given as follows \footnote{The data sampling event is not 
triggered if $\|  \varpi(t)\|   =  \|  q(t)\| = 0$.},
\begin{eqnarray}
t\r_{k+1}=\inf\{t>t\r_k\vert \;\|  \varpi(t)\| -\sigma(\|  q(t)\| )= 0\}\label{regutrigge}
\end{eqnarray}
with $\sigma\in\cal{K}$ satisfying
\begin{eqnarray}
\gamma(\sigma(s)) =  c s,\; 0< c <1, \;\;\forall s\geq 0. \label{smallgaincond}
\end{eqnarray}
Then, Problem 2 is solved in the sense of  \begin{eqnarray}
\lim_{t\to\infty}\|  e_{[t,\infty)}\| \leq\hat{\gamma}\Big(\lim_{t\to\infty}\|  \mu_{[t,\infty)}\| \Big) \label{reguerror}
\end{eqnarray}
for some $\hat{\gamma}\in\mathcal{K}$.  Moreover, 
\begin{eqnarray}
\inf_{k\in\mathbb{S}\r}\{t\r_{k+1}-t\r_k\}>0.\label{zeno0}
\end{eqnarray}
\end{theorem}
 
\begin{proof} First, a sufficiently smooth function $\kappa$ can be explicitly constructed following 
the recursive technique given in  \cite{chen2016robust} such that the closed-loop system
\EQQ
\dot{\xi}(t)=&f_c(\xi(t),\nu(t))=f(\xi(t),\nu(t))+E(w)\kappa(\bar{x}(t))
\ENN
is globally asymptotically stable. Hence, the system (\ref{newsys_pertub}) satisfies the 
 ISS property in (\ref{eqn16a}).
Because $q(t)$ can be rewritten as a function of $\bar x$ and $\dot {\bar x}$, and 
$\bar x$ is a component of $\xi$,  it is straightforward to prove that 
the ISS property implies the  IOS property in (\ref{eqn16a}).
The proof for that the function $\gamma$ is locally Lipschitz also follows the explicit construction of 
the ISS gain function.  More details can also be referred to in  \cite{Khan2018}.

Next, we aim to prove that the closed-loop system (\ref{newsys_pertub}) has the following ISS property
\begin{eqnarray}
\lim_{t\to\infty}\| \xi_{[t,\infty)}\| \leq\bar{\gamma}(\lim_{t\to\infty}\|  \bar{\mu}_{[t,\infty)}\| )\label{issregu}
\end{eqnarray}
for some $\bar{\gamma}\in\mathcal{K}$.

With the definition of $q(t)$, the signal $\varpi(t)$  specified by (\ref{newsignals}) is calculated
as follows:
\begin{eqnarray*}
\varpi(t)=- \int_{t_k\r}^{t}q(s)ds, t\in[t\r_k,t\r_{k+1}).
\end{eqnarray*}
As a result, the closed-loop system can be regarded as the interconnection of the $\xi$-subsystem and the $\varpi$-subsystem. In particular, the event-triggered condition in (\ref{regutrigge}) is designed such that
\begin{eqnarray}
\|  \varpi(t)\| \leq\sigma(\|  q(t)\| ),\;\;\forall t\in[t\r_k,t\r_{k+1}).\label{eqn22}
\end{eqnarray}

By substituting (\ref{eqn22}) into the IOS property of (\ref{eqn16a}), we obtain
\EQ
\|  q(t)\| \leq{\rm max} \{\beta(\|  \xi(0)\| ,0),{\gamma}(\sigma(\|  q_{[0,t]}\| )), \\
{\gamma}_{\bar{\mu}}(\| \bar{\mu}_{[0,t]}\| )\}.\label{eqn25}
\EN
From (\ref{smallgaincond}), it is seen that $\gamma(\sigma(\|  q_{[0,t]}\| ))<\|  q_{[0,t]}\| $. This implies
\begin{eqnarray}
\|  q(t)\|  \leq{\rm max}\{\beta(\|  \xi(0)\| ,0),{\gamma}_{\bar{\mu}}(\| \bar{\mu}_{[0,t]}\| )\},\;\;\forall t\geq 0. \label{eqn26}
\end{eqnarray}
Hence, from the ISS property of (\ref{eqn16a}) and (\ref{eqn22}), we have
\EQQ
\|  \xi(t)\| 
\leq{\rm max}\{\tilde{\beta}(\|  \xi(0)\| ,0),\tilde{\gamma} (\sigma(\|  q_{[0,t]}\| )),\tilde{\gamma}_{\bar{\mu}}(\| \bar{\mu}_{[0,t]}\| )\} \nonumber\\
\leq{\rm max}\{\tilde{\beta}(\|  \xi(0)\| ,0),\tilde{\gamma} (\sigma(\beta(\|  \xi(0)\| ,0))), \\
\tilde{\gamma} (\sigma ({\gamma}_{\bar{\mu}}(\| \bar{\mu}_{[0,t]}\| ))), 
\tilde{\gamma}_{\bar{\mu}}(\| \bar{\mu}_{[0,t]}\| )\},\;\forall t\geq 0.\label{ISS} 
\ENN
From above, the states $q(t)$ and $\xi(t)$ are bounded for $t\geq 0$. In particular, 
denote   $\|  q(t)\|  \leq q_\infty$ and $\|  \xi(t)\|  \leq \xi_{\infty}$ for two constants $q_\infty$ and $\xi_\infty$. 
 
\medskip

Next, we will examine the system behaviors during two intervals $[t^{\ast}/2, t^{\ast}]$ and $[t^{\ast}/4, t^{\ast}]$ for any time $t^{\ast} >0$. First, the IOS property of (\ref{eqn16a}) with $\tau_0= t^{\ast}/4$ implies
\EQ
\|  q_{[t^{\ast}/2,t^{\ast}]}\| \leq{\rm max}\{\beta(\|  \xi(t^{\ast}/4)\| , t^{\ast}/4),{\gamma}(\|  \varpi_{[t^{\ast}/4,t^{\ast}]}\| ), \\{\gamma}_{\bar{\mu}}(\| \bar{\mu}_{[t^{\ast}/4,t^{\ast}]}\| )\}.\label{eqn30}
\EN
For any $t \in [t^{\ast}/4, t^{\ast}]$, there exists an integer $k$ such that $t \in [t\r_k, t\r_{k+1})$.   Then, from the event-triggering mechanism (\ref{eqn22}), 
\begin{eqnarray}
\| \varpi_{[t^{\ast}/4,t^{\ast}]}\| \leq\sigma(\|  q_{[t^{\ast}/4,t^{\ast}]}\| ).\label{eqn31}
\end{eqnarray}
Substituting (\ref{eqn31}) into (\ref{eqn30}) gives
\EQ
\|  q_{[t^{\ast}/2,t^{\ast}]}\| \leq & {\rm max}\{\beta(\|  \xi(t^{\ast}/4)\| , t^{\ast}/4), \\ & {\gamma}(\sigma(\|  q_{[{t^{\ast}}/{4},t^{\ast}]}\| )),
{\gamma}_{\bar{\mu}}(\| \bar{\mu}_{[t^{\ast}/4,t^{\ast}]}\| )\}.\label{eqn32}
\EN
 Also, it is noted that
\begin{eqnarray*}
\gamma(\sigma(\|  q_{[t^{\ast}/2,t^{\ast}]}\| ))<\|  q_{[t^{\ast}/2,t^{\ast}]}\| ,\;\;\forall \|  q_{[t^{\ast}/2,t^{\ast}]}\| >0.
\end{eqnarray*}
Then, (\ref{eqn32}) reduces to
\begin{eqnarray}
\|  q_{[t^{\ast}/2,t^{\ast}]}\| \leq{\rm max}\{\beta(\xi_{\infty}, t^{\ast}/4),{\gamma}(\sigma(\|  q_{[t^{\ast}/4,t^{\ast}/2]}\| )),\nonumber\\
{\gamma}_{\bar{\mu}}(\| \bar{\mu}_{[t^{\ast}/4,t^{\ast}]}\| )\}.\label{eqn33}
\end{eqnarray}
Denote $z(t^{\ast}):=\|  q_{[t^{\ast}/2,t^{\ast}]}\|$.  Then,
one has $z(t)\leq q_\infty, \forall t\geq 0$ and
$ z(t^{\ast}/2)  =\|  q_{[t^{\ast}/4,t^{\ast}/2]}\| $.  
As a result, (\ref{eqn33}) can be rewritten as
\EQ
z(t^{\ast})\leq{\rm max} \{\beta(\xi_{\infty},t^{\ast}/4), \gamma(\sigma(z(t^{\ast}/2))),  \\ 
{\gamma}_{\bar{\mu}}(\| \bar{\mu}_{[t^{\ast}/4,t^{\ast}]}\| )\} , 
\forall t^* >0.\label{inez}
\EN
 
Next, we will show that for any $\delta>0$, there exists $T>0$, such that
\EQ
z(t^{\ast})\leq{\gamma}_{\bar{\mu}}(\lim_{t\to\infty}\| \bar{\mu}_{[t,\infty]}\| )+\delta,\;\forall t^*>T. \label{inez1}
\EN
Otherwise, there exists a positive $\delta$, such that, for any $T$, there exists $t^{\ast}> {T}$ such that, 
\EQ
z(t^{\ast})>{\gamma}_{\bar{\mu}}(\lim_{t\to\infty}\| \bar{\mu}_{[t,\infty]}\| )+\delta\geq\delta.\label{inez2}
\EN
Pick a positive integer $n$ satisfying 
\EQQ q_\infty<\frac{\delta}{c^n}\ENN
and $T$ satisfying 
\begin{eqnarray*}
\beta(\xi_{\infty},\frac{T}{4^n})<\delta
\end{eqnarray*}
and
\begin{equation*}
{\gamma}_{\bar{\mu}}(\| \bar{\mu}_{[\frac{T}{4^n},\infty]}\| )<{\gamma}_{\bar{\mu}}(\lim_{t\to\infty}\| \bar{\mu}_{[t,\infty]}\| )+\delta.
\end{equation*}
So, there exists $t^{\ast}>T$ such that $z(t^{\ast})>\delta$. From (\ref{inez}), since
\EQQ
z(t^{\ast})&>\delta>\beta(\xi_{\infty},\frac{T}{4^n})>\beta(\xi_{\infty},\frac{t^{\ast}}{4^n})\nonumber\\
z(t^{\ast})&>{\gamma}_{\bar{\mu}}(\lim_{t\to\infty}\| \bar{\mu}_{[t,\infty]}\| )+\delta\\
&\geq{\gamma}_{\bar{\mu}}(\| \bar{\mu}_{[\frac{T}{4^n},\infty]}\| )\geq{\gamma}_{\bar{\mu}}(\| \bar{\mu}_{[\frac{t^{\ast}}{4},t^{\ast}]}\| ),
\ENN
one has 
\begin{eqnarray*}
z(t^{\ast})\leq \gamma(\sigma(z(t^{\ast}/2))) = c z(t^{\ast}/2).
\end{eqnarray*}
By repeating this manipulation $n$ times, one has
\begin{eqnarray*}
\delta<z(t^{\ast}) \leq c^n z(\frac{t^{\ast}}{2^n})
\end{eqnarray*}
and hence
\begin{eqnarray*}
q_\infty<\frac{\delta}{c^n}<z(\frac{t^{\ast}}{2^n})
\end{eqnarray*}
 which is a contradiction with $z(t)\leq q_\infty, \forall t\geq 0.$ 

\noindent Finally, (\ref{inez1}) implies 
\begin{equation*}
\lim_{t\to\infty}\|  q_{[t,\infty)}\| \leq\gamma_{\bar{\mu}}(\lim_{t\to\infty}\| \bar{\mu}_{[t,\infty]}\| )
\end{equation*}
and hence
\begin{equation*}
\lim_{t\to\infty}\|  \varpi_{[t,\infty)}\| \leq\sigma(\gamma_{\bar{\mu}}(\lim_{t\to\infty}\| \bar{\mu}_{[t,\infty]}\| )).
\end{equation*}
It is ready to see that (\ref{issregu}) holds. 

Since $e=\bar x_1$ is a component of $\xi$ and 
 $\bar{\mu}=[\tilde{p}_0(\nu),\cdots, \tilde{p}_r(\nu),\tilde{f}_1(\nu),0,\cdots,0]\t \mu$
with $\nu={\rm col}(v,w)$  bounded,  (\ref{issregu}) implies (\ref{reguerror})
and Problem 2 is thus solved.

The remaining proof for the avoidance of Zeno behavior (\ref{zeno0}) is referred to in the proof of \cite{Khan2018}. 
As the proof is irrelevant to the additional $\bar{\mu}$ in the system, it can be directly used here and the details are thus ignored. 
 \end{proof}

%
%


\section{A Numerical Example}

In this section, we demonstrate the event-triggered output synchronization of four heterogeneous nonlinear MASs.  Consider four nonlinear agents described by the following lower triangular structure: 
 \EQ
\dot{z}_i&=-z_i+\begin{bmatrix}0\\2\end{bmatrix}x_i \\
\dot{x}_i&=-\begin{bmatrix}0&1\end{bmatrix}z_i+\begin{bmatrix}1&0\end{bmatrix}z_ix_i+w_ix_i+u_i \\
y_i&=x_i,\;\;i=1,\cdots,4,\label{agent12}
\EN
where $x_i\in \mathbb{R}$, $z_i\in\mathbb{R}^2$, and  $\vert w_i\vert\leq 1$. The objective is to synchronize all the agents' outputs to a sinusoidal pattern governed by (\ref{pattern}) with:
\begin{eqnarray*}
A=\begin{bmatrix}0&-1\\1&0\end{bmatrix}, \;\;c(v_{\infty}(t))=\begin{bmatrix}1&0\end{bmatrix}{v}_{\infty}(t).
\end{eqnarray*}    
Consider a network equipped with a directed graph whose asymmetric Laplacian matrix $\mathcal{L}$ is given as 
\begin{eqnarray*}
\begin{bmatrix}1& 0& 0& -1\\-1& 1& 0& 0\\ 0& -1 &1& 0\\ 0& 0& -1& 1\end{bmatrix}.
\end{eqnarray*} 
The positive left eigenvector associated with zero eigenvalue of $\mathcal{L}$ is given as $\mathbf r=\begin{bmatrix}1/4&1/4&1/4&1/4 \end{bmatrix}\t$. Choose $g_1=g_2=g_3=g_4=1$, $\lambda=0.19$, and $\beta=2.5$. Then on solving (\ref{eqnr}) with $B=[0,1]\t$, we obtain $P=\begin{bmatrix} 6.07&-1.12\\-1.12&5.00\end{bmatrix}$ and subsequently $K=\begin{bmatrix}-1.12&5.00\end{bmatrix}$, $b_1=11.25$ and $b_2=10.25$. By selecting $\eta_1=\eta_2=\eta_3=\eta_4=0.0425$, $\eta=0.045$, and $\phi=0.03$, we obtain a lower bound $b=0.0542$, as defined in  (\ref{lowerbound}).   Employing
the reference model (\ref{exosys})  with the control law 
(\ref{eventcont})
and the event-triggering mechanism (\ref{maxtrigger}), we achieve consensus of reference models in the sense of (\ref{eqn3}) as shown in Fig. \ref{config}. The Zeno behavior is excluded, which is demonstrated by the event-triggered sampling intervals in Fig.~\ref{ref2}.
This verifies solvability of Problem 1 as expected by Theorem~\ref{T1:2}.

Next, we consider an explicit solution to Problem 2, that is the 
the event-triggered perturbed output regulation problem for each agent and its associated reference model. 
As mentioned in section \ref{regulation}, the agent label $i$ is ignored for notation conciseness, as 
the design is given in a completely distributed fashion. 
Assumption~\ref{RA1} is trivially satisfied for the system under consideration. 
Let $v=[v_1,v_2]\t$.  Assumption \ref{RA2} is satisfied with the following solution to the regulator equations,
\begin{eqnarray*}
\mathbf{z}(v,w)&=&\begin{bmatrix}0\\v_1+v_2\end{bmatrix}\nonumber\\
\mathbf{x}(v,w)&=&v_1\nonumber\\
\mathbf{u}(v,w)&=&(1-w)v_1.\nonumber
\end{eqnarray*} 
Then, the steady-state generator is given by
\begin{eqnarray*}
\dot{\theta}(v,w)=T\Phi T^{-1}\theta(v,w),\;\mathbf{u}(v,w)=\Psi T^{-1}\theta
\end{eqnarray*}
with the pairs ($\Psi,\Phi$)
\begin{eqnarray*}
\Psi=\begin{bmatrix}1&0\end{bmatrix}, \Phi=\begin{bmatrix}0&-1\\1&0\end{bmatrix}.
\end{eqnarray*}
Choose a pair of controllable matrices 
\begin{eqnarray}
M=\begin{bmatrix}-1&0\\0&-2\end{bmatrix},\;N=\begin{bmatrix}1\\2\end{bmatrix}
\end{eqnarray}
and solve $T$ from the Sylvester equation $MT+N\Psi=T\Phi$.
It verifies Assumption~\ref{RA3} and gives the dynamic actuator compensator 
\EQ
u&=\bar{u}+\Psi T^{-1}\eta \\
\dot{\eta}&=M\eta+Nu.\label{internalexample}
\EN
Next, performing the coordinate transformation (\ref{coordinate}) on the system composed of (\ref{agent12}) and (\ref{internalexample}) leads to
\EQ
\dot{z}_{0}=&-{z}_{0}+\begin{bmatrix}0\\2\end{bmatrix}e-\begin{bmatrix}0\\v_1-v_2\end{bmatrix}B\mu\nonumber\\
\dot{z}_{1}=&Mz_{1}+\begin{bmatrix}0&1\end{bmatrix}Nz_{0}-\begin{bmatrix}1&0\end{bmatrix}Nz_0e-\begin{bmatrix}1&0\end{bmatrix}Nv_1z_0\\
&+(MN-wN)e-(Nv_2+T\Phi ^{-1}\theta(v,w))B\mu\\
\dot{e}=&-\begin{bmatrix}0&1\end{bmatrix}z_{0}+\begin{bmatrix}1&0\end{bmatrix}z_0e+\begin{bmatrix}1&0\end{bmatrix}v_1z_0+we\\
&+\Psi T^{-1}z_{1}+\Psi T^{-1}Ne+\bar{u}+v_2B\mu.
\EN
It is easy to verify Assumption~\ref{RA4} for the above system. Thus, 
the event-triggered controller is explicitly calculated as follows
\begin{eqnarray*}
\bar{u}(t)=\kappa(e(t_k\r))=-30e(t_k\r)-e^3(t_k\r), \;t\in[t_{k}\r,t_{k+1}\r).
\end{eqnarray*}
With the explicit function $\kappa$, the signals $\varpi(t)$ and $q(t)$ can be computed, 
as well as the the IOS gain function $\gamma(s) =40s$.
As a result,  the event-triggering mechanism (\ref{regutrigge}) is implemented with $\sigma(s)=(0.99/40)s$, satisfying (\ref{smallgaincond}). The simulation result is illustrated in Figs. \ref{simuresult} and \ref{trg2}. It is observed that the robust output synchronization in a sinusoidal pattern is achieved.  Also, the sampling intervals verify avoidance of Zeno behavior.
\begin{figure}[!t]
\centering
\includegraphics[width=3in]{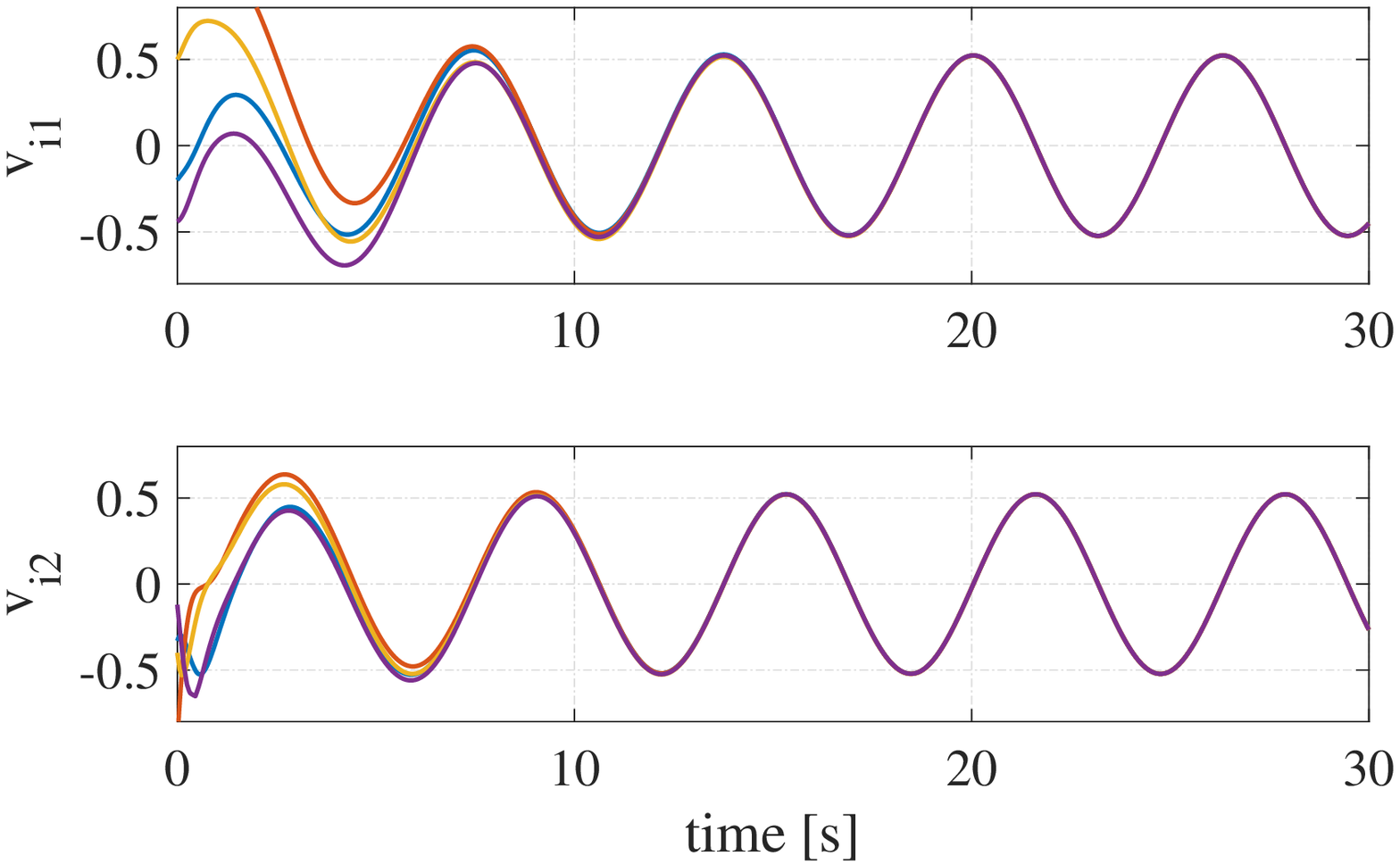}
\caption{Consensus of reference models, $i=1,\cdots,4$.\label{config}}
\centering
\includegraphics[width=3.in]{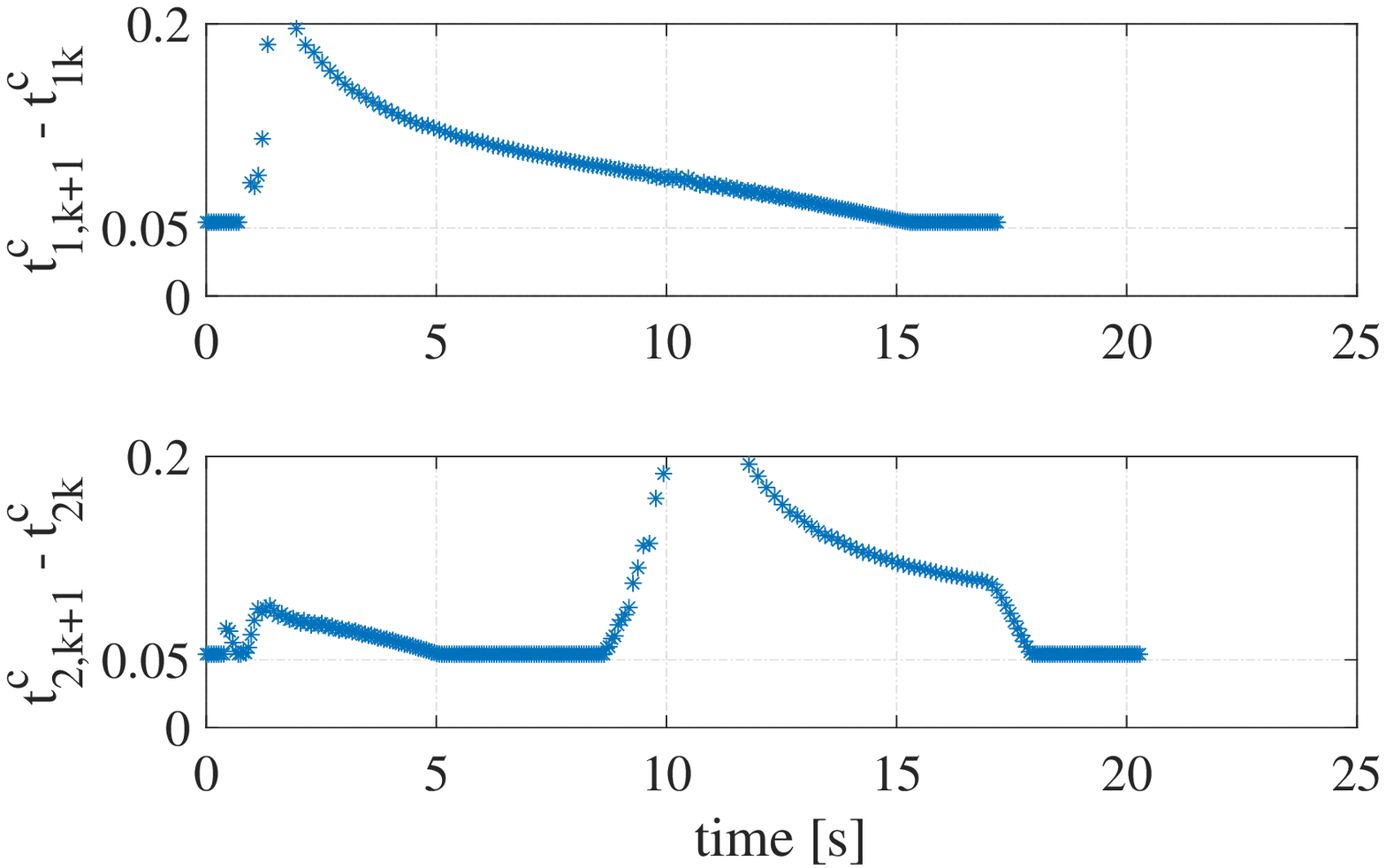}
\includegraphics[width=3in]{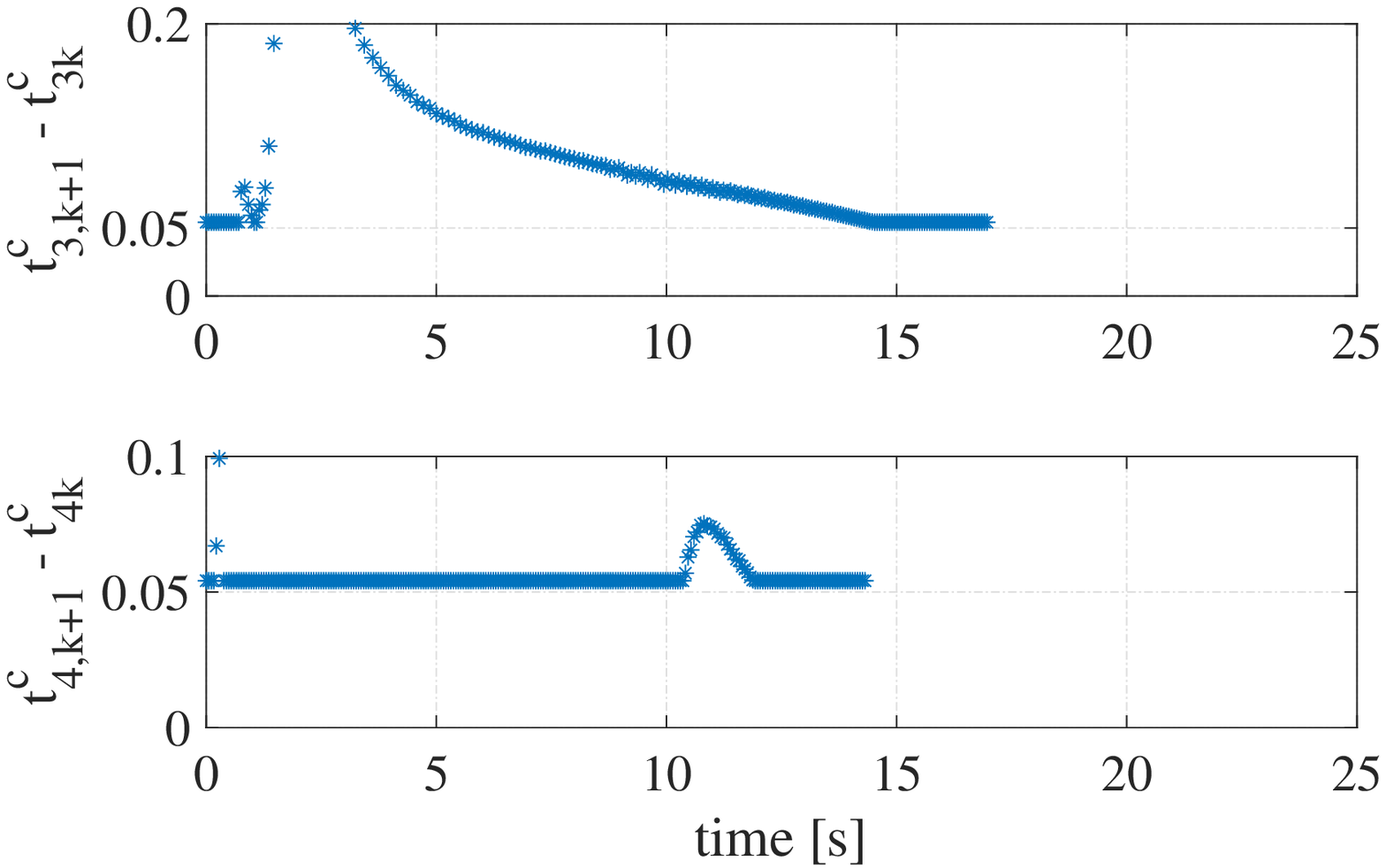}
\caption{Event-triggered sampling intervals in consensus of reference models.\label{ref2}}
\end{figure}

\begin{figure}[!t]
\centering
\includegraphics[width=2.8in]{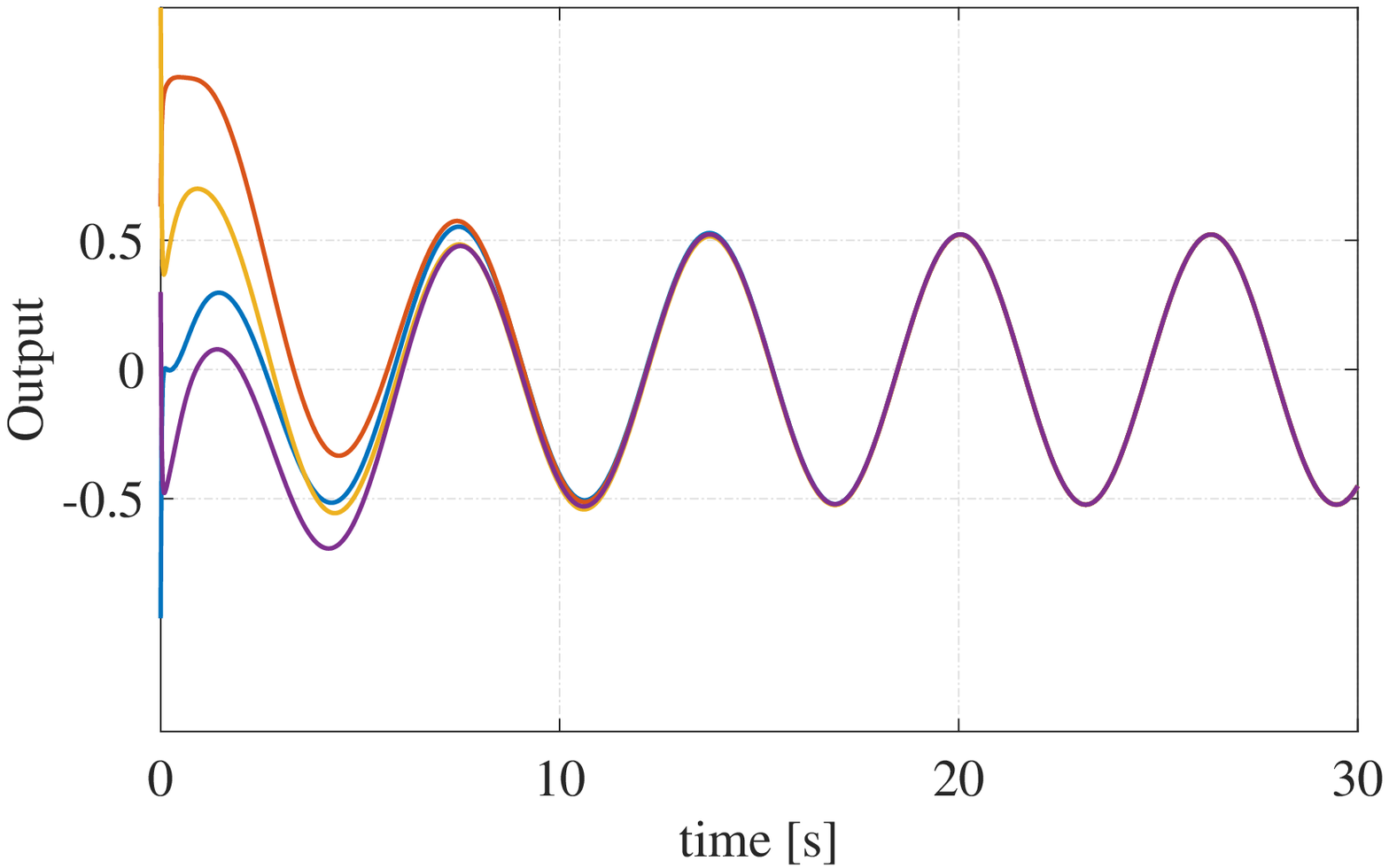}
\caption{Synchronization profiles of the four agents.\label{simuresult}}
\centering
\includegraphics[width=3in]{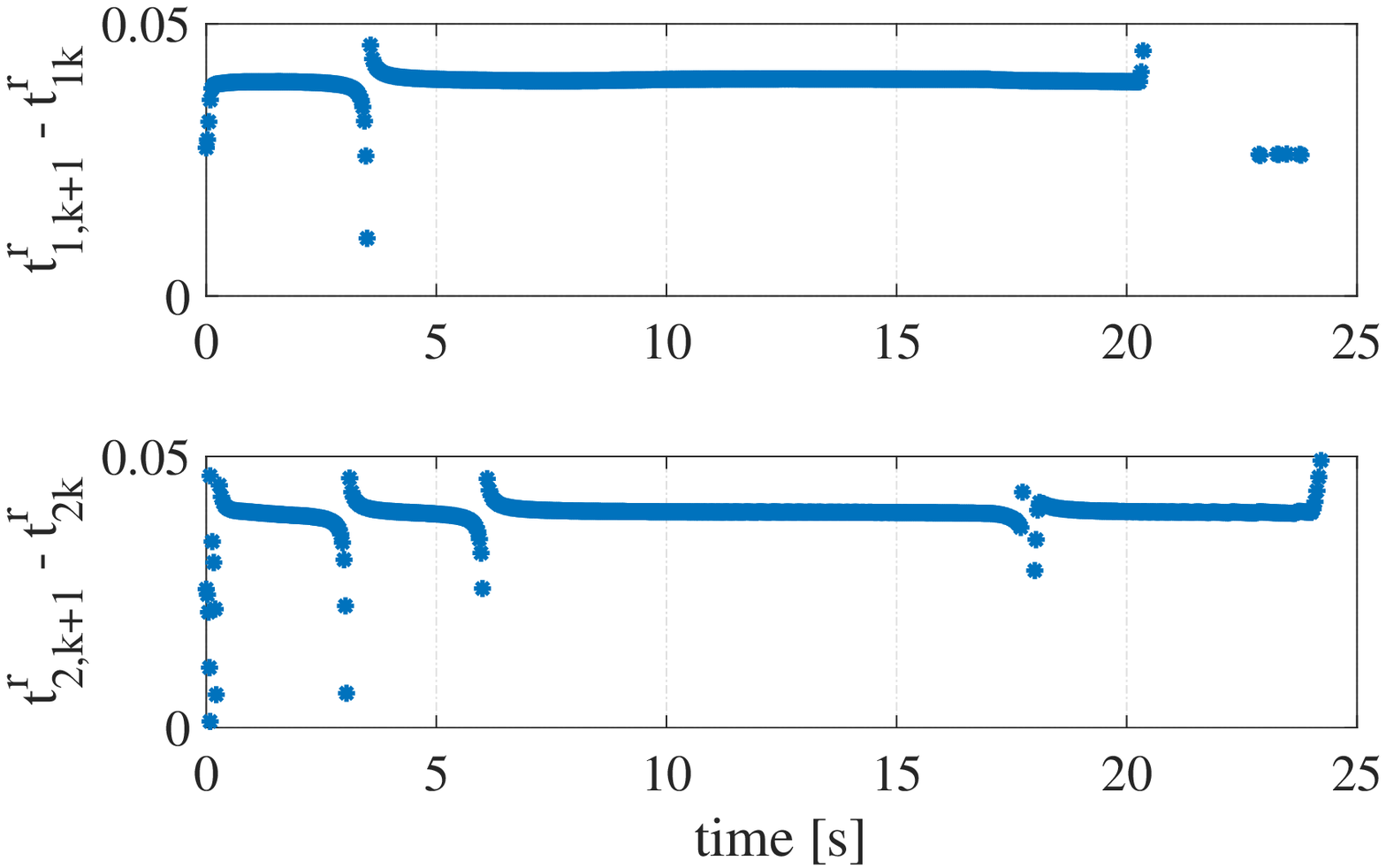}
\includegraphics[width=3in]{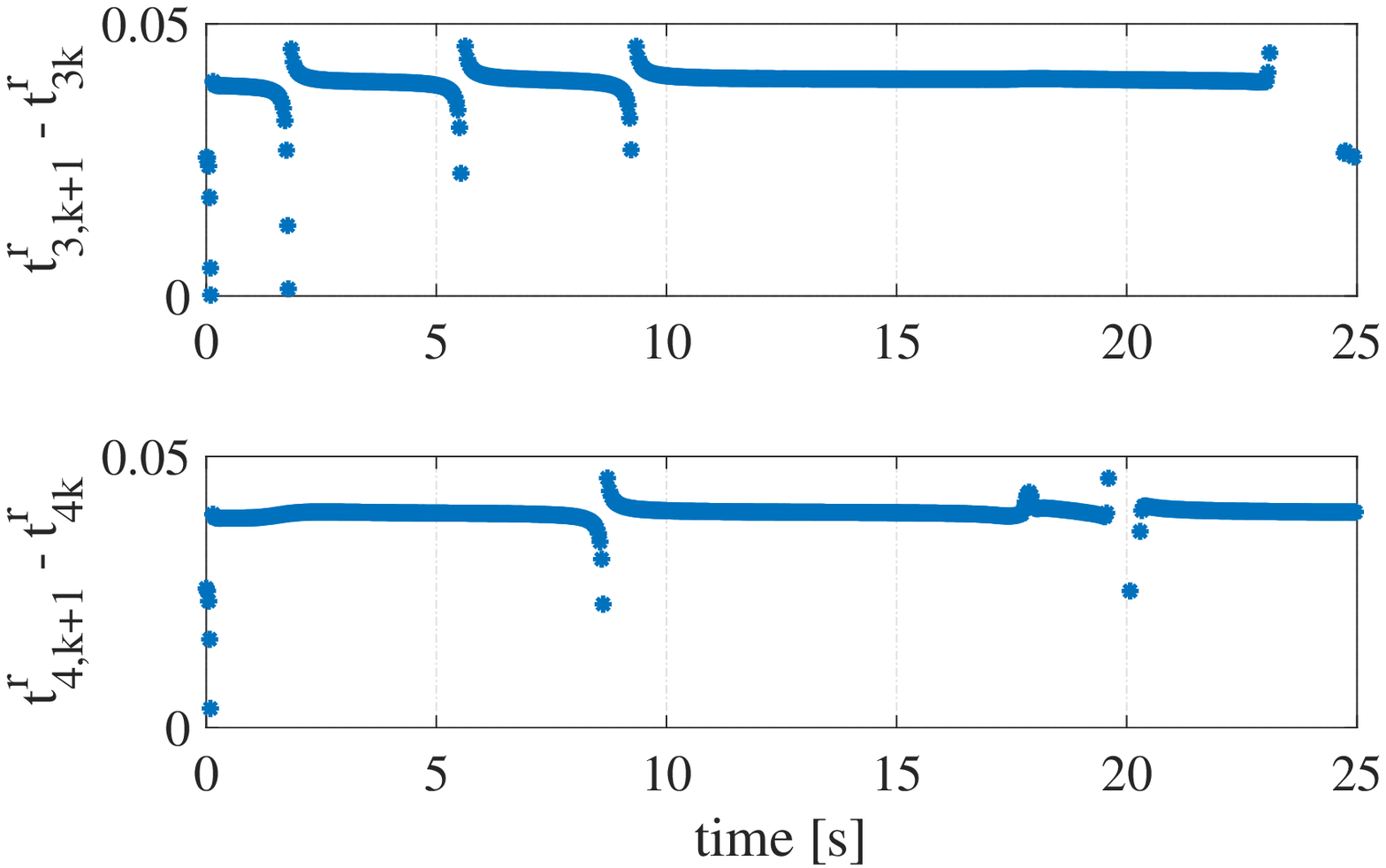}
\caption{Event-triggered sampling intervals in output regulation.\label{trg2}}
\end{figure}

\section{conclusion}\label{con}
In this paper, we have proposed an explicit solution to the event-triggered output synchronization problem for
a class of heterogeneous nonlinear multi-agent systems in a network with a directed topology. 
It has also been shown that continuous monitoring  among agents and Zeno phenomenon can be avoided. 
The two-fold result given in this paper has contributions to development of networked event-triggering mechanisms 
for both linear systems and nonlinear systems. 

 
\bibliographystyle{IEEEtran}
\bibliography{eventbib}

\end{document}